\RequirePackage{fix-cm}
\documentclass[10pt]{article}

\usepackage{cite}
\usepackage{amsmath,amssymb}
\usepackage{amsthm}
\usepackage{microtype}
\usepackage{framed}
\usepackage{fullpage}
\usepackage{array}
\usepackage{graphicx}
\usepackage{float}
\usepackage{url}
\usepackage{hyperref}
\usepackage{multirow}
\usepackage{color,soul}
\usepackage{xspace}
\usepackage{tabularx}
\usepackage{subcaption}

\usepackage{enumitem}
\usepackage{algorithm}
\usepackage{framed}
\usepackage[noend]{algpseudocode}

\algblockdefx{AtState}{END}[1]{\underline{In state {\sf #1}:}}{}
\algblockdefx{AtStateBreakLine}{END}[2]{ \underline{In state {\sf #1}} \newline  \underline{{\sf #2}:} }{}
\algcblockdefx{AtState}{AtState}{END}[1]{\underline{In state {\sf #1}:}}{\vspace{-1.5em}}{}
\algcblockdefx{AtStateBreakLine}{AtStateBreakLine}{END}[2]{ \underline{In state {\sf #1}} \newline  \underline{{\sf #2}:} }{}

\newcommand{\thickhline}{%
    \noalign {\ifnum 0=`}\fi \hrule height 1pt
    \futurelet \reserved@a \@xhline
}
\newcommand{\Explore}{\textsc{Explore}\xspace}
\newcommand{\LExplore}{\textsc{LExplore}\xspace}
\newcommand{\sInitL}{\textsf{InitL}\xspace}
\newcommand{\sInit}{\textsf{Init}\xspace}
\newcommand{\sHappy}{\textsf{Happy}\xspace}
\newcommand{\sForward}{\textsf{Forward}\xspace}

\newcommand{\sBounce}{\textsf{Bounce}\xspace}

\newcommand{\sReverse}{\textsf{Reverse}\xspace}
\newcommand{\sReturn}{\textsf{Return}\xspace}
\newcommand{\sTerminate}{\textsf{Terminate}\xspace}
\newcommand{\sFComm}{\textsf{FComm}\xspace}
\newcommand{\sBComm}{\textsf{BComm}\xspace}
\newcommand{\sFirstBlockL}{\textsf{FirstBlockL}\xspace}
\newcommand{\sFirstBlock}{\textsf{FirstBlock}\xspace}

\newcommand{\sAtLandmarkL}{\textsf{AtLandmarkL}\xspace}
\newcommand{\sAtLandmark}{\textsf{AtLandmark}\xspace}
\newcommand{\sReady}{\textsf{Ready}\xspace}
\newcommand{\pFailed}{\textsl{failed}\xspace}
\newcommand{\aLook}{\textsf{Look}\xspace}
\newcommand{\aMove}{\textsf{Move}\xspace}
\newcommand{\aCompute}{\textsf{Compute}\xspace}
\newcommand{\pMeeting}{\textsl{meeting}\xspace}
\newcommand{\pCatches}{\textsl{catches}\xspace}
\newcommand{\pCatched}{\textsl{caught}\xspace}
\newcommand{\dLeft}{\textit{left}\xspace}
\newcommand{\dRight}{\textit{right}\xspace}
\newcommand{\vTtime}{\ensuremath{Ttime}\xspace}
\newcommand{\vTsteps}{\ensuremath{Tsteps}\xspace}
\newcommand{\vTnodes}{\ensuremath{Tnodes}\xspace}
\newcommand{\vEtime}{\ensuremath{Etime}\xspace}
\newcommand{\vEsteps}{\ensuremath{Esteps}\xspace}
\newcommand{\vBtime}{\ensuremath{Btime}\xspace}
\newcommand{\vNtime}{\ensuremath{Ntime}\xspace}

\newtheorem{observation}{Observation}
\newtheorem{lemma}{Lemma}
\newtheorem{corollary}{Corollary}
\newtheorem{theorem}{Theorem}

\begin{document}

\title{Live Exploration of  Dynamic Rings}


\author{G. Di Luna\footnote{University of Ottawa, gdiluna@uottowa.ca, g.a.diluna@gmail.com}, S. Dobrev \footnote{Slovak Academy of Sciences, stefan.dobrev@savba.sk}, 
P. Flocchini \footnote{University of Ottawa, flocchin@site.uottawa.ca}, N. Santoro \footnote{Carleton University, santoro@scs.carleton.ca}}


\maketitle

\begin{abstract}

In the    {\em graph exploration} problem, a team of mobile computational entities, called agents, arbitrarily positioned at some nodes of a graph,
must  cooperate  so that each node is
eventually visited by at least one agent.
 In the  literature, the  main focus has been on graphs that are {\em static}; that is,
 the topology is either invariant in time or subject to localized  changes.

 The few studies  on exploration of {\em dynamic} graphs have been almost all limited to the  {\em centralized}  case
(i.e., assuming  complete a priori  knowledge of  the  changes and the times of their occurrence).

 We investigate   the {\em decentralized}   exploration of   dynamic graphs  (i.e.,  when the agents  are unaware
of the  location and timing of the changes) focusing, in this paper,  on   dynamic systems whose underlying graph is a  {\em ring}.

We first  consider  the  {\em fully-synchronous} systems  traditionally assumed in the literature; i.e., all agents  are active at  each  time step.
We then   introduce  the notion of  {\em semi-synchronous} systems,   where only a subset of  agents might be active at each time step  (the   choice of the subset  is  made by an adversary); this model is common in the context of  mobile agents in continuous spaces  but has never been studied before for agents moving in graphs.
  Our main focus is on the impact that the level of synchrony as well as other factors such as    anonymity,  knowledge of the size of the ring,  and chirality (i.e.,  common orientation) have on the solvability of the problem, focusing on the minimum number of agents necessary.

  We draw an extensive map of feasibility,  and  of  complexity  in terms of   minimum
number of agent movements. All our sufficiency proofs are constructive, and almost all our solution protocols are asymptotically optimal.


\end{abstract}

\section{Introduction}

\subsection{Framework}

\subsubsection{Graph Exploration}

The  problem  of {\em graph exploration}  requires a team of mobile computational entities, usually called {agents} or {robots},   located at the nodes of a graph and capable of moving from node to neighbouring node, to explore the graph,
with the requirement that each node is
eventually visited by at least one agent.

This  classical  problem has been   extensively investigated, starting from the pioneering work of Shannon \cite{Sh51}.
In the vast literature on the subject (e.g., see \cite{AlbH97,DaFKNS07,DengP99,DieP14,FlKMS09,FraIPPP04,PaP99}),  a wide spectrum of different assumptions have been made and
examined e.g. with regard to:   the computational power of the agent(s);  the  structure of the graph  and its
properties;  the level of topological knowledge available to the agents;
whether or not  the network is anonymous (i.e., the nodes lack distinct identifiers); whether the nodes can be marked (e.g., by leaving a pebble); the level of
synchronization. In case of multiple agents, different types of means of communication  have been considered, including:  face-to-face,  where agents communicate when they are on the same node;  wireless, where the  agents are able to communicate even when they are on different nodes;  with whiteboard,  where the agents communicate by writing on and reading from  whiteboards present in each node.

Regardless of their differences, all these investigations  share the common assumption that the graph is  {\em  static}: its  topological structure does not change during the exploration.
This is true also for those investigations   that  consider faulty nodes and/or faulty links (e.g., see \cite{BaDFS15,CoKR06,DoKS13,DoFPS06}).

\subsubsection{Dynamic Graphs}

In distributed computing,  researchers  recently started to investigate
 highly {\em  dynamic graphs}, that is graphs where the topological changes are not
 localized and sporadic; on the contrary,  the topology  changes continuously and at unpredictable locations.
These investigations have been motivated by  the
 development of
 highly dynamic networks, where
changes are not anomalies (e.g., faults) but rather integral part of the nature of the system.
These highly dynamic networks are modelled
in a natural way in terms of  {\em time-varying graphs},
a model  formally defined in \cite{CaFQS12},
where
main classes of systems studied in the literature  and their computational relationship were identified.
If time is discrete (e.g., changes occur in rounds), the evolution of these systems can
 be  equivalently described as a sequence of static graphs,  called {\em evolving graph},  a model suggested in \cite{HarG97} and formalized  in \cite{Fe04}.

The study of distributed computations in
highly dynamic graphs
has focused mainly on problems of information diffusion
and reachability, such as broadcast, routing, etc. (e.g., see
\cite{AwE84,BCF09,BFJ03,BreDKV13,CaFMS14,CaFMS15,CleMPS11,GodM14}),
and on problems of coordination and agreement, such as election, consensus, etc. (e.g., see
\cite{ArGSS13,AuPR13,BRS12,DilBBC14,KLyO10,KuMO11}).
Clearly,  all these studies make  strong assumptions in order to  restrict the
 universe of the possible topological changes and their temporal occurrence. One such restriction is that  topological changes are {\em periodic}  (Class 8 of \cite{CaFQS12}),
 such as in carrier graphs  (e.g., see \cite{BreDKV13,FlKMS12a,FlMS13,IlW11,LW09b}).
 A popular restriction is by assuming that the network is {\em always connected} (Class 9 of \cite{CaFQS12}): at each time instant,  there is a connected spanning subgraph;
further assuming that such connected spanning subgraph persists for $T\geq 1$ time units
defines the extensively studied sub-class of {\em T-interval-connected} systems (e.g., see \cite{di2015brief,KLyO10,KuMO11,OdW05}).

\subsubsection{Exploration of Dynamic Graphs}

Returning to the exploration problem, very little is known in the case of dynamic graphs. On the probabilistic side, there is an early seminal work on random walks \cite{AvKL08}.
 On the deterministic side there are:  the   study of
 the complexity of computing a foremost exploration schedule  under the  {1-interval-connectivity} assumption \cite{Michail2014},  generalized and extended
 in \cite{ErHK15};
  the computation of an exploration schedule  for {\em rings}  under the  {T-interval-connectivity} assumption \cite{IlW13}; and the computation of an exploration schedule
 for {\em cactuses} under the  {1-interval-connected} assumption \cite{IlKW14}.
All  these studies are however mainly {\em centralized} (or off-line,  post-mortem); that is, they assume that the exploring agents have complete
{a priori} knowledge of  the topological changes and the times of their occurrence.

Very little is known on the {\em distributed} (or on-line,  live) case, i.e. when the location and timing of the changes are unknown to the agents.
Exploration of    {\em carrier networks}, a  periodic  class of time varying graphs, by a synchronous agent, has been studied in   \cite{FlMS13}, where the feasibility of the problem is investigated  depending on the knowledge available to the agent (size of the network or upper bound, length of the period)  and where  optimal solutions are proposed. Under a slightly different model, similar results can be found also in \cite{IlW11}.   Exploration  has  also been examined   assuming that the dynamic graph is $\delta$-recurrent
(i.e., each edge appears at least once every $\delta$ rounds)  \cite{IlW13}.
Apart from these results, to the best of our knowledge, nothing  is known.

\subsection{Contributions}

 In this paper, we  investigate   the distributed exploration of
 dynamic rings under the  1-interval connectivity assumption by mobile agents without explicit means of communication.

\begin{table*}
\footnotesize
\begin{tabularx}{\textwidth}{ |c|X|X|X| }
 \hline
 	{\em N. Agents} & {\em Assumptions} &  {\em Even if} & {\em Result}  \\
 \hline  \hline

 	2 &  No knolwedge on $n$, No Landmark & Non-anonymous agents, Chirality&      Partial Termination {\bf Impossible} (Th. \ref{impossibilitywithIDterminationsync}) \\ \hline
	Any &  No knolwedge on $n$, No Landmark,  Anonymous agents & Chirality &      Partial Termination {\bf Impossible} (Th. \ref{impossibilitymoreagents}) \\
 \hline
\end{tabularx}
\caption{Impossibility results for ${\cal FSYNC}$ model\label{table:fsynimp}.}
\end{table*}

\begin{table*}
\footnotesize
\begin{tabularx}{\textwidth}{ |c|X|X| }
 \hline
 	{\em N. Agents} & {\em Assumptions} &  {\em Exploration with  Termination}  \\
 \hline  \hline

 	2 &  Known  bound  $N$ & Explicit Termination in time  $3N-6$ (Th. \ref{tmh-knownNNoChirality})\\ \hline
	 	\hline
 	2 & Chirality, Landmark  & Explicit Termination in time  $O(n)$ (Th. \ref{thm:unknownNWithChirality}) \\
 	 	 	 	\hline
 	2 &  Landmark &  Explicit Termination in time  $O(n \log(n))$ (Th. \ref{thm:LandmarkNoChilarity}) \\

 \hline
\end{tabularx}
\caption{Possibility results for   ${\cal FSYNC}$ model\label{table:fsyn}.}
\end{table*}

 We consider three  different termination requirements once the ring has been explored: {\em Explicit Termination}, where,  within finite time, all agents must explicitly terminate and stop moving;
 {\em Explicit Partial Termination}, where, within finite time,    at least one agent explicitly terminates and stops moving;
 {\em Unconscious Exploration}, where the agents are not required to stop.

We first  consider  the  {\em fully-synchronous} systems   (Section \ref{sync})  traditionally assumed in the literature; i.e., all agents  are active at  every  time step.
We then   introduce  the notion of  {\em semi-synchronous} systems (Section \ref{ssync}),   where at each time step  only a subset of  the agents might be active  (the   choice of the subset  is  made by an adversary). The semi-synchronous  model is common in the context of  mobile agents in continuous spaces (e.g., \cite{book}) but has never been studied before for agents moving in graphs.
  Our main focus is on the impact that the level of synchrony as well as other factors such as    knowledge of the size of the ring, chirality (i.e.,  common sense of orientation),
   and  anonymity, have on the solvability of the problem.




We  start by examining the exploration problem in  {\em fully synchronous}   systems (${\cal FSYNC}$),  with  {\em two} agents, after showing that it is unsolvable with  only one.

 For {\em anonymous} rings, we first show that, without any additional knowledge on the ring size (e.g., an upperbound),
 two non-anonymous agents cannot explore and partially terminate.
 If the agents are anonymous, the same impossibility holds for any number of agents.
 On the other hand, unconscious exploration is possible using two anonymous agents and no additional knowledge.

If there is knowledge of an upper bound $N$ on the ring size, we show that two anonymous agents can explore and they can both terminate in $3N-6$ rounds. This can be done even if agents do not have a common chirality.

For  {\em non-anonymous} rings,  we show that the presence of a single observably different node ({\em landmark}) allows
 two anonymous agents to solve the  exploration problem with explicit  termination. This can be done without the need of any  additional information.
 We provide an algorithm that terminates in $O(n)$ rounds when there is chirality, and an $O(n \log (n))$ algorithm for the case without chirality.
 For the case of no chirality we do not know if $O(n \log (n))$ rounds are necessay.
 \color{black}


 A summary of these results is shown
in Tables    \ref{table:fsynimp} and  \ref{table:fsyn}.
 All the sufficiency proofs of ${\cal FSYNC}$ are constructive and, apart for the algorithm for the case of {\em non-anonymous} ring and no chirality, all  the proposed algorithms are asymptotically {\em optimal}. %

%

\begin{table*}
\footnotesize
\begin{tabularx}{\textwidth}{  |m{0,8cm}|c|m{2cm}|m{5cm}|X| }
 \hline
 	{\em Model} &{\em N. Agents} & {\em Assumptions} &  {\em Even if} & {\em Result}  \\
 \hline  \hline
 	({\sf NS})& Any & None 		&  Chirality,  Known $n$,  Landmark, Non-anonymous Agents &  Exploration  {\bf impossible} (Th. \ref{nsimpossible}) \\  \hline
	\hline
	\multirow{2}{*}{({\sf PT})}& 2 &No Chirality, Anonymous Agents  &     Known  $n$, Landmark &  Exploration  {\bf impossible} (Th. \ref{twonotenoughpt})\\ \cline{2-5}
	  & 2 & None  & Chirality, Known  $n$, Landmark &    Explicit Termination   {\bf impossible} (Th. \ref{parttermination}) \\ \hline
	  \hline
	  ( {\sf ET})& Any & Unknown $n$ & Known   bound $N$, Chirality,  Landmark, Non-anonymous Agents & Partial Termination {\bf impossible} (Th. \ref{ETnottermination})\\
 \hline
\end{tabularx}
\caption{Impossibility results in ${\cal SSYNC}$ models\label{table:ssynimp}.}
\end{table*}

 \begin{table*}
 \footnotesize
\begin{tabularx}{\textwidth}{ | c | c | X | X | }
 \hline
 	{\em Model} & {\em N. Agents} & {\em Assumptions} &  {\em Exploration and Termination}  \\
 \hline
 \hline

 	\multirow{4}{*}{({\sf PT})} & \multirow{2}{*}{2}  & Chirality, Known   bound $N$&   Partial Termination   in  $O(N^2)$ moves (Th. \ref{PTupperbound})
	  \\
	   \cline{3-4}
 &  & Chirality, Landmark&   Partial Termination   in  $O(n^2)$ moves  (Th. \ref{PTlandmark})\\

 	 \cline{2-4}
	& \multirow{2}{*}{3}  &  Known   bound $N$ & Partial Termination   in  $O(N^2)$ moves (Th. \ref{PTupperboundNoChirality}) \\
	  	 \cline{3-4}
		&  &Landmark  & Partial Termination   in  $O(n^2)$ moves (Th. \ref{PTupperboundNoChiralityLandmark}) \\
 	 \hline
	  \hline
	 \multirow{2}{*}{({\sf ET})}& 2 &    Chirality &   Unconscious  Exploration possible  (Th. \ref{th:perpetual})\\
           		 \cline{2-4}
	 & 3 &  Known $n$ &   Partial Termination possible (Th. \ref{th:ETpartialtermination})  \\

 \hline
\end{tabularx}
\caption{Possibility results for  ${\cal SSYNC}$ models\label{table:ssyn}.}
\end{table*}

  We then examine  the problem in  {\em semi-synchronous}  systems (${\cal SSYNC}$). In these systems it is  possible that
an  agent $a$ waiting to traverse a  missing link $e$  is inactive  in the round the edge reappears.  Depending on what happens  to  that  agent,
we consider and analyze three  different {\em transportation models}  (described in details  later in the paper) and establish feasibility and complexity
results:\\

 \indent -  {\em No Simultaneity}  ({\sf NS}) model :  $a$  is not allowed to move while inactive. This is the weakest of the models that we consider.  In this case, exploration  is impossible with any number of agents, even with exact  knowledge of the ring size,  nodes with distinct IDs,  and common chirality.

\indent -  {\em Passive Transport} model  ({\sf PT}):  $a$ is transported on $e$. In this case,
we show  that, without chirality, two anonymous robots are not sufficient to explore the  ring; the result holds even if  there is a distinguished landmark node and the exact network size is known. On the other hand, with chirality, two agents  can perform the exploration if there is a known upper bound   on the ring size or there is a landmark node.
As for termination, we show that it is impossible to guarantee explicit termination of both agents (even if exact knowledge of the size, chirality and landmark are available). On the other hand, we
prove that it is always possible for at least one of the agents to terminate.
Interestingly,  presence of chirality allows to solve the problem with only $2$ agents;   without chirality  $3$ agents are necessary.
The {\sf PT} model is the strongest of the ${\cal SSYNC}$ models.

\indent -   {\em Eventual Transport}  ({\sf ET}) model:  $a$ is not allowed to  move  while inactive but,  should  the edge be present for an infinite number of rounds, $a$ is guaranteed to be eventually active at a round when the edge is present. In this case, we show that
 exploration with  partial termination  of all agents  is impossible, regardless of the number of agents, even if an upper bound on the ring size is known,   nodes have distinct IDs, and agree on chirality.  On the other hand, with exact knowledge of the ring size, we prove that exploration is possible with three anonymous agents even without chirality, and  at least one  agent explicitly terminates.\\

The results are summarized
in Tables \ref{table:ssynimp} and \ref{table:ssyn}.
 Also  for ${\cal SSYNC}$,  all the sufficiency proofs are constructive,  and almost all the proposed algorithms are asymptotically {optimal}.

\section{Model and Basic Limitations}

\subsection{{\bf Model and Terminology}}
Let  ${\cal R} = (v_0,\ldots v_{n-1})$  be
a synchronous ring where, at any time
step $t\in {\sc N}$,  one of its edges might not be present;
the choice
on which link is missing (if any) is made by an adversary.
Such a dynamic network is   known in the literature  as a {\em
1-interval connected}  ring.

 Each node $v_i$  is connected to
its two neighbours $v_{i-1}$ and $v_{i+1}$ via distinctly labeled ports
  $q_{i-}$ and   $q_{i+}$, respectively\ (all operations on
the indices are modulo $n$); the labeling of the ports may not be
globally consistent, thus might not provide an orientation, and the label may not be comparable \cite{CHALOPIN200854}.
The ring is said to be {\em anonymous} if the nodes have no distinguishable
identifiers, and {\em with landmark} if there is a node (the landmark) which is
different from all others.

Operating in  ${\cal R}$ is a  set $A=\{a_{0},\ldots,a_{m-1}\}$  of
agents, each provided with  memory and computational capabilities.
The agents are anonymous
and all execute the same protocol.
Any number of agents can
reside at a node at the same time.
   Initially located     at    arbitrary nodes,  not necessarily distinct, they do not have any explicit
communication mechanism, nor can leave  marks   on the nodes.
The agents are {\em mobile}, that is they can move from node to
neighboring node.
To move, an agent has to position itself on the port from which it
wants to leave. Access to ports is done in  mutual exclusion: an agent will not succeed to gain an already occupied port;
when several agents try to position themselves on the same port, only one of them succeeds.
Two agents moving in opposite directions on the same edge in the same
round might not be able to
detect each other.

Each agent $a_j$  has a consistent private orientation of the ring;
that is, it has a function $\lambda_j$
which designates each port either {\em left} or {\em right} and
$\lambda_j(q_{i-})=\lambda_j(q_{k-})$, for all  $0\leq i,k<n$.
The orientations of the agents might not be the same.
If all agents agree on the same orientation and are aware of it,
we say that there is {\em chirality}.

The system operates in {\em  synchronous} time steps, called {\em
rounds}.   Initially, all agents are {\em inactive}. Each time step
$t\in {\sc N}$ starts with a non-empty subset  $A(t)\subseteq A$ of the agents becoming
{\em active}.
Upon activation,   agent $a_j\in A(t)$  at node $v_i$ performs a sequence of
operations:  \aLook,  \aCompute, and (possibly) \aMove.

  \begin{enumerate}
 \item

 {\sf Look:} The agent determines its own position within the
node (i.e., whether or not it is on a port, and if so on which one),
  and the position  of the  other agents (if any) at that node. We
call this information a {\em snapshot}.

\item
 {\sf Compute:} Based on the snapshot and the content of its
local memory, the  agent   executes its  protocol (the same for all
agents) to
determine whether or not to move  and, if so, in what direction; the
result will be  $direction\in  \{left,right,nil\}$, where $left$ and
$right$ are with respect to its own local orientation.
If
$direction = nil$, the agent  becomes inactive.
  If $direction \neq  nil$, $a_j$  attempts to access the appropriate
port (if not already there);  if it gains access, it positions itself
on the port, otherwise it sets private variable  $moved=false$ and
becomes inactive.

\item
{\sf Move:} Let the agent be  positioned on port  $q_{i-}$
(resp., $q_{i+}$) after computing. If the link between $v_i$ and
$v_{i-1}$ (resp., $v_{i+1}$) is present in this round,  then
  agent $a_j$ will move to $v_{i-1}$ (resp., $v_{i+1}$), reach it,
set private variable $moved=true$, and become inactive.
  If the  link between $v_i$ and $v_{i-1}$ (resp., $v_{i+1}$) is not present,
   then   agent $a_j$ will remain  in the port, set  $moved=false$, and become
inactive. In either case, access to  port  $q_{i-}$  (resp.,
$q_{i+}$) continues to be denied to any other requesting  agent
during this round.
\end{enumerate}

By definition,  the delays are such that all active agents have
become inactive by the end of  round $t$;   the system  then starts
the new round $t+1$.

Notice that, since  access to a port is in mutual exclusion, in the same  round  at most
one agent will move in each direction on the same edge. Recall that
two agents moving in opposite directions on the same edge in the same
round might not be able to
detect each other.

A major computational factor is the nature of the  activation schedule
of the agents.
   If $A(t)=A$ for all $t\in {\sc N}$,  that is all agents are
activated at every time step, the system is said to be
     {\em fully synchronous}  (${\cal FSYNC}$). Otherwise the system
is said to be
{\em semi-synchronous}  (${\cal SSYNC}$);
the agents that are not activated in a round are said to be {\em
sleeping} or {\em passive} in that round; the choice of which agent is active in a round
is made under an adversarial scheduler, where every agent is  activated infinitely often.
 When an agent is activated, it does not know whether or not it was active in the previous round.
Observe  that  in  ${\cal SSYNC}$ it is possible for an
agent to be sleeping on a port.
This is indeed the case when  an agent $a$
gains access to a port $q$ when the link is not there (thus, it remains on $q$),
and $a$ is not activated
in the next round. What  may happen to an agent sleeping on a port
gives raise to   different models, described  in the following:

\begin{itemize}
\item{No Simultaneity ({\sf NS})}:
A sleeping agent cannot move.
There is no  guarantee of simultaneity for an agent sleeping on a port.

\item{Passive Transport ({\sf PT})}:
If an agent   is sleeping on a port at round $t$
and the corresponding edge is present in that round,   the agent   is
 moved to the other endpoint of the edge in round $t$.

  \item{Eventual Transport ({\sf ET})}:
A sleeping agent cannot move.
If an agent is sleeping on a port at round $t$ and the corresponding edge is present infinitely many times, then the agent will eventually become active at a round $t'>t$  when the  corresponding edge is present  (simultaneity condition).

\end{itemize}


The algorithm executed by the agents solves the {\em exploration} problem if,  within finite time, every node of the ring is
visited by at least one agent.  The exploration is said to be
with {\em explicit termination} if every agent  executing the algorithm  within finite time
enters a terminal state  and
no longer moves. The terminal state has to be entered only after the exploration of the ring.
The exploration is said to be with  {\em  explicit partial termination} (or just {\em partial termination}) if at least one  agent  executing the algorithm
enters a terminal state within finite time  and
no longer moves.
Finally, it is said to be {\em unconscious} if   the agents are not required to stop nor to be aware that the ring has been visited.

\subsection{{\bf Basic Limitations} \label{syncimp}}
 We begin our study by showing simple impossibility results.


\begin{observation}\label{obs-oneAgent}
The adversary can prevent an  agent from leaving the initial node $v_0$,
by always removing the edge over which the agent wants to leave $v_0$.
\end{observation}
From this Observation,  we immediately get:
\begin{corollary}\label{cor-oneAgent}
A single agent is not able to explore the ring.
\end{corollary}

Hence, at least two agents are needed. However, the adversary can prevent their cooperation:

\begin{observation}\label{obs-2AgentsCantMeet}
The adversary can prevent two agents starting at different locations
from meeting each other  even if they have unlimited memory, common chirality
and distinct known IDs.  This result holds even if the scheduler is  ${\cal FSYNC}$
and never blocks both agents at the same round.
\end{observation}
 \begin{proof}
The adversary will never remove an edge, except in the case when that would lead to agents meeting in the next step. 
There are two possible cases how the agents can meet in the next step.\\
Case 1:  One agent is waiting at a node  and the other agent, at a neighbouring node, 
decides to traverse the edge $e$ connecting the two nodes. In this case, the adversary removes edge $e$. \\
Case 2:Both agents decide to traverse different edges $e$ and $e'$ leading to the same vertex. 
Again, it is sufficient for the adversary to remove one of the two edges to prevent   rendezvous.
\end{proof}


\begin{theorem}\label{impossibilitywithIDterminationsync}
There does not exist any partially terminating deterministic exploration algorithm
of anonymous rings of unknown size by  two  agents, even with distinct IDs, common chirality,
and when the scheduler is {\cal FSYNC}.
\end{theorem}

\begin{proof}
By contradiction,  assume that there exists a terminating exploration algorithm ${\cal A}$. Let us consider an execution ${ E}$ of ${\cal A}$ on a dynamic ring of size $n$,
where agents $a$ and $b$ start in two distinct locations and where the adversary always prevents the meeting of the agents,   never blocking the agents at the same round. By Observation \ref{obs-2AgentsCantMeet} this run exists.

 Let us assume, without loss of generality,  that in such an execution agent $a$ is the first one terminating at round $r({ E})$.
 Let us now consider an execution  ${ E}'$ of ${\cal A}$ on a dynamic ring of size $n'=8r( E)$, where the agents  start  at two distinct locations at distance $4r({ E})$.
 The execution  ${ E}'$ is constructed in such a way  that, until round $r( E)$ neither agent can distinguish this execution from ${ E}$. This is possible since in $E$ the adversary never blocks the
 two agents at the same time and they do not meet.

Since the size of the ring is $8r( E)$ and the agents started $4r( E)$ apart, at round $r( E)$ (when agent $a$ terminates) the distance between the two agents is at least $2r( E)$ 
and there are at least $6r({ E})$ unexplored nodes. The execution ${ E}'$ is completed by the adversary blocking agent $b$ at round $r( E)$ and afterwards,  preventing it from exploring the unexplored nodes.
Therefore, in this execution, the partial termination of algorithm ${\cal A}$ is incorrect.
 \end{proof}

Observe now that for anonymous agents the impossibility of explicit termination holds  regardless of their number. This is because,
in the setting when there is orientation and no edge is removed,  all the agents will act in the same way at each time step. If one were to
decide to terminate after $t$ time steps, they all would do so at the same time; however, since they do not know the ring size,
they would do the same also in a ring of size
$n $ multiple of $t$,    terminating without having completed the exploration.
Indeed, with the same reasoning, it is immediate that even the weaker partial termination is impossible.

\begin{theorem}\label{impossibilitymoreagents}
There does not exist any partially  terminating deterministic exploration algorithm of anonymous rings of unknown size by anonymous  agents, regardless of their number. The result holds even if agents have common chirality, and  the scheduler is  ${\cal FSYNC}$.
\end{theorem}

Summarizing, without some knowledge of the size of the ring
 or without the asymmetry introduced by a landmark node,
 exploration with partial termination is impossible
 even in the fully synchronous model by  two agents with IDs, or  by any number of anonymous agents.
 Notice that Theorems \ref{impossibilitywithIDterminationsync} and  \ref{impossibilitymoreagents} would hold  even if the agents were equipped with  wireless communication.

As for the amount of time required for exploration, there exists the following lower bound  due to \cite{ErHK15}:

\begin{observation} \label{thm-WithChiralityLB} \cite{ErHK15}
Exploration of an anonymous ring  by  two anonymous agents
requires  at least $2n-3$ time in the worst case, even if there is chirality and   the scheduler is  ${\cal FSYNC}$.
\end{observation}

\section{Ring Exploration in ${\cal FSYNC}$ \label{sync}}

In this section,
we    consider   exploration   when the system is fully synchronous,  presenting and
analyzing protocols that solve the problem  under different assumptions on
 knowledge of the ring size, anonymity of the nodes,  and presence of chirality. All   these solutions do not require the agents to be able to communicate explicitly.

Our algorithms use   as a building block   procedure \Explore($dir$ $|$ $p_1:s_1$; $p_2:s_2$; \dots; $p_k:s_k$),
where $dir$ is either \dLeft~or \dRight,  $p_i$ is a predicate,  and $s_i$ is a state.
In Procedure  \Explore,  the agent performs \aLook, then evaluates the predicates $p_1, \ldots, p_k$ in order;
as soon as a predicate is satisfied, say $p_i$, the procedure exits and the agent does a transition to the specified state, say $s_i$. If no predicate is satisfied, the agent tries to \aMove in the specified direction $dir$ and the procedure is executed again in the next round.

Furthermore, the following variables are maintained by the algorithms:
\begin{itemize}
\item \vTtime, \vTsteps: the total number of rounds and edge traversals, respectively, since the beginning of the execution of the algorithm.
\item \vEtime, \vEsteps: the total number of rounds and edge traversals, respectively, since   the last call of procedure \Explore.
\item \vBtime: the number of consecutive rounds the agent has been currently waiting in a port.
\end{itemize}
In particular, the following predicates are used:
\begin{itemize}
\item \pMeeting: both agents are in the node.
\item \pCatches: the agent observes the other agent on the port corresponding to its moving direction.
\item \pCatched: the agent is on the port after a failed move, the other agent is observed in the node.
\end{itemize}

Observe that, in a fully synchronous system, when predicate \pCatches holds for an agent, then \pCatched holds for the other agent. In the following, we say that the agents {\em catch each other} if both predicates hold.

\subsection{{\bf Known Upper Bound on Ring Size} \label{knownupp}}
In this section we study the simple case  of exploring the ring   when the agents  know an upper-bound $N\geq n$  on the ring size,
and we show how to solve the problem in asymptotically optimal time, even without chirality.

The directions \dLeft and \dRight now refer to the local orientation of the individual   agent. The algorithm is shown in Figure~\ref{alg-knnch}. We use the predicate \textsl{failed}, that is verified when an agent tries to enter a port  and it fails to do so.

The algorithm works as follows.
At the beginning, each agent goes left; recall that the left direction could be different for the two agents. An agent keeps going left unless:
$i)$ it catches the other agent in the first $2N-4$ rounds; or
$ii)$  $2N-4$ rounds have passed and the agent has been blocked for $N-1$ rounds; or
$iii)$  it is in a node and it fails to enter a port. In all these cases, the agent changes state to \sBounce and it goes right until termination, at round $3N-6$.

If instead an agent is caught in the first $2N-4$ rounds,  then it enters in state \sForward and it keeps the left direction until termination, at round $3N-6$.

\begin{figure}

\begin{framed}
\begin{algorithmic}
\State States: \{{\sf Init}, {\sf Bounce}, {\sf Forward}, {\sf Terminate}\}.
\AtState{Init}
    \State \Call{Explore}{\dLeft$|$ ($\vTtime \geq 2N-4 \land \vBtime = N-1) \lor \pFailed$: \sBounce; \pCatches: \sBounce; \pCatched: \sForward; $\vTtime \geq 2N-4$: \sForward}
\AtState{Bounce}
    \State \Call{Explore}{\dRight$|$ $\vTtime\geq 3N-6$: \sTerminate}
\AtState{Forward}
    \State \Call{Explore}{\dLeft$|$ $\vTtime\geq 3N-6$: \sTerminate}
\END
\end{algorithmic}

\caption{Algorithm {\sc KnownNNoChirality} \label{alg-knnch}}
\end{framed}

\end{figure}

\begin{theorem}  \label{tmh-knownNNoChirality} Algorithm  {\sc KnownNNoChirality} allows  two anonymous agents
without  chirality to explore  a 1-interval connected ring and to  explicitly terminate  in time $3N-6$, where   $N$  is a known upper-bound on the ring size.
\end{theorem}
\begin{proof}
 The termination by round $3N-6$ is trivial for the condition $\vTtime\geq 3N-6$.
It is sufficient to show that the ring has been explored when $\vTtime=3N-6$.
Let $a,b$ be the two agents.
\begin{itemize}

\item We first examine the case where $a,b$ start on the same node $v$.

 If they have different agreements on the $\mathit{left}$ direction then no agent can change direction in the first $2N-4$ rounds. It is trivial to see that after at most $N-1$ rounds the agents have explored the ring, and hence the Theorem holds.

 Let us consider then the case in which the agents agree on the same $\mathit{left}$ direction, and hence try to traverse the same edge $e$. Since the access to a port is done in mutual exclusion, only one of them will enter the port, while the other will fail.
If $e$ is not missing at round $1$, then at the beginning of round $2$ the two agents will be in different nodes; otherwise they will be on the same node.
 In either case, for one of the two agents  the predicate \textsl{failed} is true; notice that the \textsl{failed} predicate will not be verified in any other case. 
 This implies that the two agents will have different directions, and they will not change such directions in the first $2N-4$ rounds; thus in the next $N-1$ rounds the ring will be explored. Hence the Theorem holds.

\begin{figure*}
\begin{center}
\includegraphics[scale=0.6]{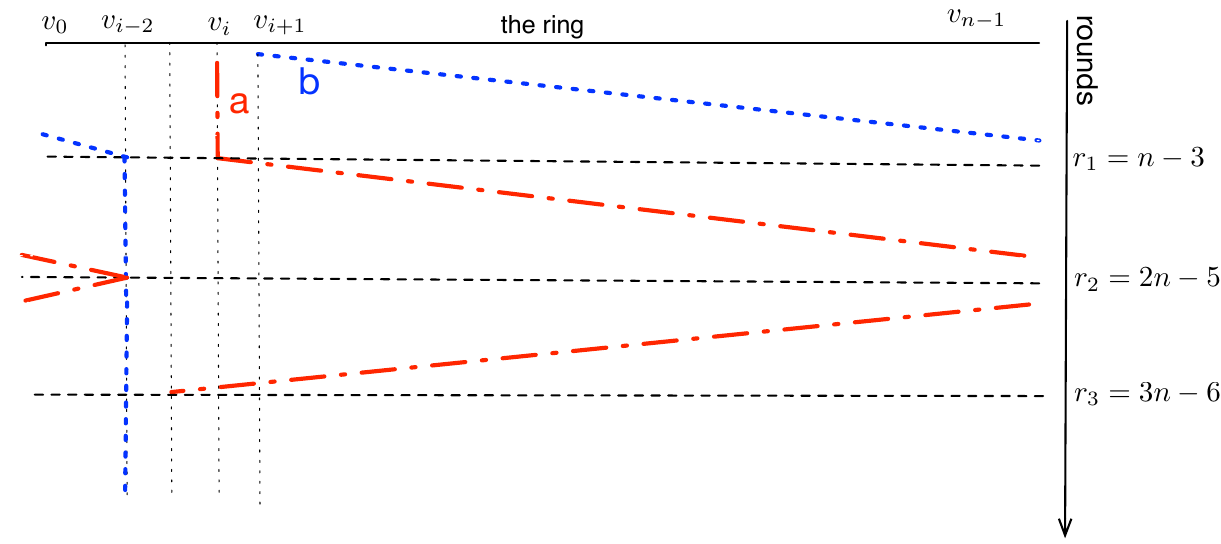}
\end{center}
\caption{Schedule where Algorithm {\sc KnownNNoChirality} takes $3n-6$ rounds to explore.\label{fig-tight1}}

\end{figure*}

\color{black}

\item We now  consider when $a$ and $b$ start on different nodes.

 Let us examine the case when they disagree on the $\mathit{left}$ direction.
 After $N-3$ rounds, they either (i)  are at distance $2$, or (ii)  are at distance $1$, or (iii) crossed each other, where the distance is the number of edges in the portion of the ring pointed by the left directions of the agents. 
 (i) If the distance between the agents is $2$, then they were initially on two neighbouring nodes; therefore, in the next round, at time $N-2$, the  ring will be explored.
(ii)  If they are at distance $1$ and the edge between them is not missing for the successive $N-1$ rounds, then they will cross each other and the ring will be explored by at most round $3N-6$.  If the edge between them is missing for the successive $N-1$ rounds, then at round $2N-4$ they will switch direction and the ring will be explored in the next $N-2$ rounds.
(iii) If  they cross  each other on edge $e$,  they will not change direction for the successive $N-1$ rounds, thus they will explore the ring.
That is, in all three cases the Theorem holds.

 Finally, let us consider   the case when they agree on the $\mathit{left}$ direction.

  If they catch each other before round $2N-4$, then the ring will be explored in the next $N-1$ rounds terminating the exploration by round $3N-6$.

 If they do not catch each other, then at each round at least one of them will traverse an edge, since they cannot be blocked at the same time. 
 Let us suppose that $a$ traverses $N-k$ edges with $k <2$; then
 $b$ traversed at least $N-4+k \geq N-1$ edges, exploring the ring. 
 So the last remaining case is when agent $a$ traverses exactly $N-2$ edges; but this implies that also $b$ has traversed at least $N-2$ edges and, since they start from different nodes, the ring has been explored also in this case.

\end{itemize} 

\end{proof}

\noindent  Notice that there exists a schedule in which the exploration takes $3n-6$ rounds, so the cost is tight for $N=n$. The schedule is reported in the Figure \ref{fig-tight1}: the agents start on distinct nodes, $a$ on $v_i$ and $b$ on $v_{i+1}$, and there is  chirality. Agent $a$ is  blocked for the first $n-3$ rounds; in the meanwhile agent $b$ reaches node $v_{i-2}$.
At this point, $b$ is blocked on node $v_{i-2}$ until round $r_2=2n-5$; during this time $a$ moves until it  catches $b$; this happens exactly at round $r_2$. In the next $n-1$ rounds, $b$ is still blocked while $a$ reaches node $v_{i-1}$, exploring the ring.

The algorithm is asymptotically optimal, as shown by the following theorem, even with the lesser requirement of partial termination.

\begin{theorem} \label{thm-knownNWithChiralityLB}
Exploration with partial termination of an anonymous ring  by  two anonymous agents
with knowledge of an upper bound $N$ on the ring size
requires  at least $N-1$ time in the worst case, even with chirality.
\end{theorem}

\begin{proof}
By contradiction, let ${\cal A}$ be a correct exploration algorithm that allows  partial termination 
with knowledge of an upper bound $N$  in   at most $N-2$ time unit.
Let  $R_i$ be an anonymous ring  of size  $i$, and for $N>5$ let $R(N)=\{R_i : 3\leq i\leq N\}$.
Consider now the simultaneous execution of ${\cal A}$ in $R(N)$ started by initially placing in each $R\in R(N)$
 two agents with chirality  at two neighbouring nodes,  and then  making no edge ever disappear. 
 
Clearly, at each time step, all the agents in all  rings of $R(N)$ have  the same view, perform exactly the same movement  in the same direction;
furthermore, in each ring they are
unaware of each other, and keep their distance to $1$. 
By symmetry, if one of the agents terminates in  a ring at time $t$, then they all do in all the rings at the same time.
By assumption, the execution of ${\cal A}$ terminates at a time $t\leq N-2$; this means that in $R_{N}$, when  both agents  terminate,
 at least one node is still unexplored at that time, contradicting the correctness of ${\cal A}$.
\end{proof}

\subsection{{\bf No Known Bounds On Ring Size} \label{noknowledge}}

We now consider exploring the ring when no upper-bound  on its size is available to the agents.
Under this condition, by Theorem \ref{impossibilitywithIDterminationsync}, it is {\em impossible} for two agents to explore an anonymous ring with termination, even if the agents have unique IDs.
Hence, for exploration to occur, either termination must not be required  or the ring must not be anonymous.
In the following we consider precisely those two cases. We first show how unconscious  exploration can be performed without
any other condition even if the agents are anonymous. We then consider a  ring in which there is a special node, called landmark, different from the others and visible to the agents;  we  prove that  exploration   can be performed with termination,  even if the agents are anonymous, in time $O(n)$ if there is chirality, $O(n \log n)$ otherwise.

\subsubsection{ Unconscious  Exploration \label{oblivious:exploration}}

We present a protocol, {\sc Unconscious Exploration}, that allows two anonymous  agents to perform exploration without knowing any bound on the ring size.
The basic idea of the algorithm is for each agent  to   guess the size of the ring with an initial estimate $G$
and move in one direction for  a time equal to twice the estimate;
the agent will then  double the size estimate. It changes direction if it has been blocked for a time that is equal to the previous estimate and it will keep direction otherwise. This process is repeated with the new guess.  The algorithm is shown in Figure~\ref{alg-PerpExlp}.

\begin{figure}

\begin{framed}
\begin{algorithmic}
\State States: \{{\sf Init}, {\sf Bounce}, {\sf Reverse}, {\sf Forward}, {\sf Keep}\}.
\AtState{Init}{}
    \State $G\gets 2$, $dir\gets\dLeft$
    \State \Call{Explore}{$dir$; $\vEtime\geq 2G \land \vBtime > G$: \sReverse, $\vEtime\geq 2G$: {\sf Keep}, \pCatches: \sBounce, \pCatched:\sForward}
\AtState{Reverse}
    \State $F\gets 2 \cdot G$, $dir\gets opposite(dir)$
    \State \Call{Explore}{$dir$;  $\vEtime\geq 2G \land \vBtime > G$: \sReverse, $\vEtime\geq 2G$: {\sf Keep}, \pCatches: \sBounce, \pCatched:\sForward}
    \AtState{Keep}
        \State $G\gets 2 \cdot G$
    \State \Call{Explore}{$dir$;  $\vEtime\geq 2G \land \vBtime > G$: \sReverse, $\vEtime\geq 2G$: {\sf Keep}, \pCatches: \sBounce, \pCatched:\sForward}
\AtState{Bounce}
    \State \Call{Explore}{$opposite(dir)$}
\AtState{Forward}
    \State \Call{Explore}{$dir$}
\END
\end{algorithmic}
\caption{Algorithm  {\sc Unconscious Exploration}  \label{alg-PerpExlp}}

\end{framed}
\end{figure}

\begin{theorem}\label{thm-PerpExpl}
Algorithm {\sc Unconscious Exploration} allows  two anonymous agents
without  chirality to explore, without terminating,  an 1-interval connected ring; the exploration is completed in $O(n)$ time.
\end{theorem}
\begin{proof} 
If the agents catch each other, then they start moving 
in opposite directions and, in the subsequent  $n-1$ moves  (unknown to them), they will explore the whole ring, proving the Theorem.

Consider now the case when the agents never catch each other. 
Let a {\em phase} be the period of time when the guess remains the same. Since $G$ is always doubled after $2G$ time steps, at time $t_n\leq 4n$, $G \geq n$. Let $r$ be the first round of the phase $P$ in which $G\geq n$.

 If, at $r$, the agents are moving in the same direction, since they do not catch each other and in each time step at least one of them makes progress, in the next $2G$ time steps the ring will be explored.
 
Consider now the case when agents are moving in opposite directions:
If at time $r+2G-1$ neither of them is blocked on an edge, then by the next phase $P+1$ starting at round $r+2G$, they must have crossed each other and they keep different directions; thus, the ring will be explored by the end of phase $P+1$.
Otherwise,  at time $r+2G-1$ at least one of them has to be blocked on an edge $e$.
At the beginning of the next phase, round $r+2G$, three cases are possible:
\begin{itemize}
\item Only one agent changes direction; in this case, they will have the same direction in phase $P+1$, thus exploring in $O(n)$ rounds.
 \item Both agents reverse direction; this happens only if both have been blocked for the last $G$ rounds of phase $P$, thus they will change directions from two different endpoints of the same edge $e$, and the ring is explored by    time $r+2G+4G$.

 \item No one changes direction; this implies that by the next $n$ rounds, i.e. by round $r+2G+2G$, they either crossed each other, and thus they will explore the ring by the end of phase $P+1$, or they are blocked on the endpoints of the same edge $e'$. In the latter case, if the edge $e'$ is removed for the last $2G$ rounds of phase $P+1$, then at round  $r+2G+4G$ they both change directions moving from the two endpoints of $e'$, and thus exploring the ring by the end of phase $P+2$. Otherwise, if the edge is present for one round in the the last $2G$ rounds of phase $P+1$, then they will cross each other and continue in the same direction for the phase $P+2$, exploring the ring.
 \end{itemize}
\end{proof}

\subsubsection{Termination: Landmark and Chirality} \label{chirality}

We now focus on solutions with termination. By Theorem  \ref{impossibilitywithIDterminationsync}, in absence of bounds on the ring size, the ring cannot be anonymous. Hence, we assume  that  there is a special node, called landmark, different from the others and visible to the visiting agents. 
in this  subsection we consider  chirality, but no other   additional knowledge, and we   show that two anonymous agents can explore the ring and terminate
  in optimal time $O(n)$.


Let  $v^*$ be the  landmark, identifiable by the agents:
when performing a {\sf Look}  operation at some node $v$,   a flag  $IsLandmark$ is   set to  $true$
if and only if $v=v^*$. 
The basic idea is  to explore the ring using the landmark to compute the size and allow termination. In order to coordinate termination, the agents implicitly ``communicate"  when they catch each other (by waiting at the node if not sure whether to terminate, and by leaving it if they already know that the ring is explored). When the agents catch each other for the first time, they break symmetry and  assume different roles.

\begin{figure*}
\footnotesize

\begin{framed}
\begin{algorithmic}
\State States: \{{\sf Init}, {\sf Bounce}, {\sf Return}, {\sf Forward}, {\sf Terminate}, {\sf BComm}, {\sf FComm}\}.
\AtState{Init}
    \State \Call{LExplore}{\dLeft$|$ $\vNtime>2\ size$: \sTerminate; \pCatches: \sBounce; \pCatched: \sForward}
\AtState{Bounce}
    \State \Call{LExplore}{\dRight$|$ \pMeeting: \sTerminate; $\vEtime > 2\vEsteps \vee \vNtime >0$: \sReturn, \pCatches: \sBComm}
\AtState{Return}
    \State $\mathit{bounceSteps} \gets \vEsteps$
    \State \Call{LExplore}{\dLeft$|$ $\vNtime>3\ size \vee \pCatched$: \sTerminate; \pCatches: \sBComm}
\AtState{Forward}
    \State \Call{LExplore}{\dLeft$|$ $\vNtime \geq 7\ size \vee \pMeeting \vee \pCatches$: \sTerminate; \pCatched: \sFComm }
\AtState{BComm}
    \State $\mathit{returnSteps}\gets \vEsteps$
    \If{$\mathit{returnSteps} \leq 2 \cdot \mathit{bounceSteps}$} \Comment{both must have waited on the same edge}
        \State \aMove(\dRight) \Comment{signal the need to terminate}
        \State \sTerminate in the next round
    \ElsIf{you know that the ring is explored ($n$ is known)}
           \State \aMove(\dRight) \Comment{signal the need to terminate}
        \State \sTerminate in the next round
    \Else
        \State Stay for one round in the node
        \If {agent F is in the node} \Comment{agent F waited to learn whether to terminate}
            \State change state to \sBounce and process it (in the same round)
        \Else   \Comment{ agent F left, or tried to leave and is on the port -- signalling to terminate}
            \State \sTerminate
        \EndIf
    \EndIf
\AtState{FComm}
    \If {you know that the ring is explored ($n$ is known)}
        \State \aMove(\dLeft) \Comment{signal to B that F knows $n$}
        \State \sTerminate in the next round
    \Else
        \State \aMove from the port to the node \Comment{i.e. staying at the same node}
        \If {agent B is in the node} \Comment{this happens next round}
            \State Change state to \sForward and process it (in the same round)
        \Else \Comment{B has left or is on the port}
            \State \sTerminate
        \EndIf
    \EndIf
\END
\end{algorithmic}
\caption{Algorithm {\sc LandmarkWithChirality} \label{alg-land}}
\end{framed}
\end{figure*}

 We assign  to them  logical names: $F$  for the agent being caught, and $B$ for the one that caught  $F$. These names do not change afterwards, even though it is possible for $F$ to catch $B$ later on.

Procedure \LExplore is very similar to \Explore with the following additions:
\begin{itemize}
\item Each agent keeps track of whether it is crossing the landmark and in which direction; furthermore, it tracks its distance from the landmark (since encountering it for the first time). In this way, it can detect whether it made a full loop around the ring. When it does so for the first time,   variable $size$ is set to the ring size $n$  ($size$ is initialized to infinity, all the tests using it while it has this initial value will fail).
\item An additional variable \vNtime is maintained, tracking the total number of rounds since the agent learned $n$.
\end{itemize}


The complete pseudocode is shown in Figure~\ref{alg-land}. Both agents start going left. If they never meet,  they terminate (see Lemma \ref{lm:noCatch}). If they catch each other, the naming is done.
After naming,   agent $F$   keeps going left.
 Agent $B$  moves  right until either it completes a loop of the ring or  it is blocked  for a number of rounds larger than the number of edges it has traversed so far, i.e. predicate ``$\vEtime > 2\vEsteps$".
 When one of these conditions is satisfied, agent $B$ goes left, and it tries to catch up with $F$.

 If they catch up and $F$ has done less than $\vEsteps$ steps to the left from its old position,  then $B$ and $F$ have waited on the same edge, and hence the ring has been explored;  $B$ can detect this,  and it ``communicates" the end of exploration to $F$.

  If they catch up but $F$ has done at least $\vEsteps$ steps to the left from its old position ($B$ can detect this), they both keep executing the algorithm. Note that $F$ has made progress towards completing a ring loop, a condition that can be detected because of the landmark. Should such an event occur, then both $F$ and $B$ will eventually terminate.

  If they do not catch up  for a certain number of rounds, then they will both know that the  ring is explored and they can terminate independently (see Lemma \ref{lm:termCorr}).

\begin{lemma}\label{lm:noCatch}
In Algorithm  {\sc LandmarkWithChirality},  if the agents do not catch each other and stay in the \sInit state, then they will explore the ring and explicitly terminate by round $7n-1$.
\end{lemma}
\begin{proof}

Let the agents do not catch each other and stay in the \sInit state.

We first show that two agents starting on different nodes will explore the ring and explicitly terminate by round $7n-2$. Since the agents are moving in the same direction  but they start from different nodes, in each round at least one of them makes progress. Since they do not catch each other, the difference between the number of successful moves by the agents is at most $n-1$. Therefore, if by round $5n-2$ no agent has terminated, then both agents have crossed at least $2n-1$ edges and hence they both know $n$. By  construction, in additional $2n$ steps, the agents will terminate.
If an agent has terminated at round $r<5n-2$, this means that at time $r-2n$ this agent knew $n$, i.e. it has entered the landmark for the second time. As the agents did not catch each other, the other agent must have already entered the landmark. Since in the subsequent  $2n$ steps the agents do not catch each other and together made progress at least $2n$ times, by round $r$ the other agent will enter the landmark for the second time, and by round $r+2n$ it will terminate as well.

We examine now the case of the agents starting at the same node; since the agents are going in the same direction, at the first round they will both try to enter the same port.
 Since access to a port is granted in mutual exclusion, only one agent will succeed. If the edge $e$ corresponding to that port is missing then, at the beginning of next round, one agent sees the other inside the port after the failed move, and the agent in the port sees the other in the node. That is, the agents catch each other, contradicting the hypothesis.

Therefore, edge $e$ has to be present; this means that at the beginning of the next round the two agents will be in two different nodes.

But this is the same to consider as the agent are starting on a different nodes, hence we just have to add an additional round to the bound shown for the case of agents starting from distinct nodes.

\end{proof}

\begin{lemma}\label{lm:termCorr}
In Algorithm  {\sc LandmarkWithChirality}, if  an agent terminates, then the ring has  been explored and the other agent will terminate as well.
\end{lemma}
\begin{proof}
We prove the lemma by case analysis on how the agents terminate.
 In the proof we assume that the left direction corresponds to a counter-clockwise direction.
Consider first the case when the agents  terminate at the same time. According to the algorithm, this happens    only in  the four cases considered below;  in each case, we show  that  the ring has been explored:

\begin{enumerate}
\item Agent $F$, in state \sForward, catches agent $B$ in state \sReturn at node $w$ (see Figure \ref{figure:fforwardbreturn}):
Consider the  node $v$ where $B$ caught $F$ and changed state to \sBounce the last time. The counter-clockwise segment from $v$ to $w$ has been explored by $F$, that never changed its direction.
Consider now the node $z$ where $B$ changed state to \sReturn the last time; it is not difficult to see that $z$ must be in the counter-clockwise segment from $v$ to $w$; this in turn implies that $B$ has explored the clockwise segment from $v$ to $w$. Thus the entire ring has been explored.

\begin{figure}
\begin{center}
\includegraphics[scale=0.49]{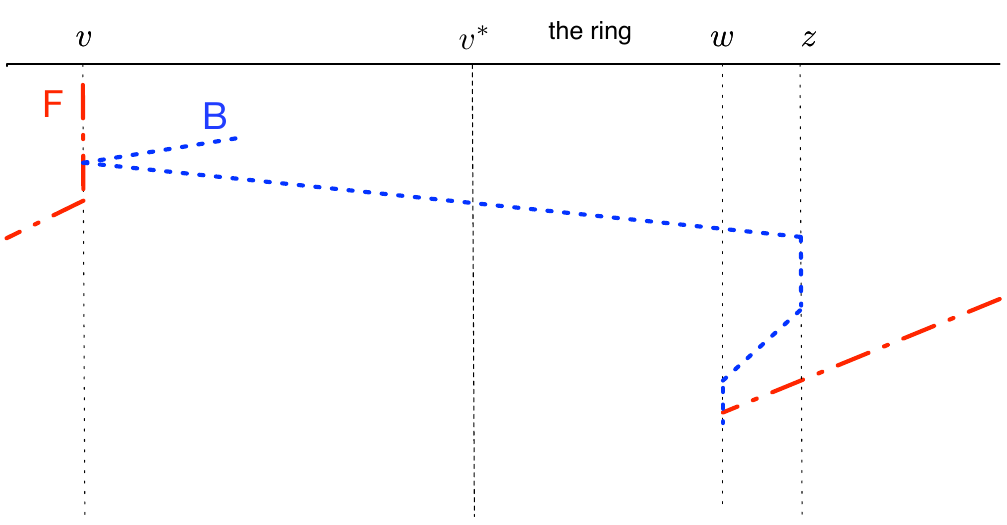}
\end{center}
\caption{ Agent $F$, in state \sForward, catches agent $B$ in state \sReturn at node $w$. \label{figure:fforwardbreturn}}

\end{figure}

\item The agents, $F$ in state \sForward and $B$ in state \sBounce, moving in opposite direction meet at a node: The entire ring has clearly been explored.

\item Agent $ F$,  in state {\textsf FComm}, knows   $n$ and signals $B$ to terminate: By construction.

\item Agent $B$, in state {\textsf BComm}, signals $F$ to terminate:
According to the algorithm when this occurs either $B$ knows $n$, or $\mathit{returnSteps} \leq 2 \cdot \mathit{bounceSteps}$. If $B$ knows $n$ the lemma trivially holds. Consider the second condition. When agent $B$ changes its direction, it has been blocked (not necessarily on the same edge) more than $\mathit{bounceSteps}$ times. Satisfying the test $\mathit{returnSteps} \leq 2 \cdot \mathit{bounceSteps}$ means  $F$ has either made progress of at most $\mathit{bounceSteps}$, or it has made one or more whole loops and then at most $\mathit{bounceSteps}$. In the first case $F$ has been blocked during one of the rounds when  $B$ has been blocked; this can only happen if they had been blocked on the same edge, i.e. the ring has been explored (see Figure \ref{figure:bsignaltermination} for an example). In the latter case, the ring has obviously been explored.

\begin{figure}
\begin{center}
\includegraphics[scale=0.5]{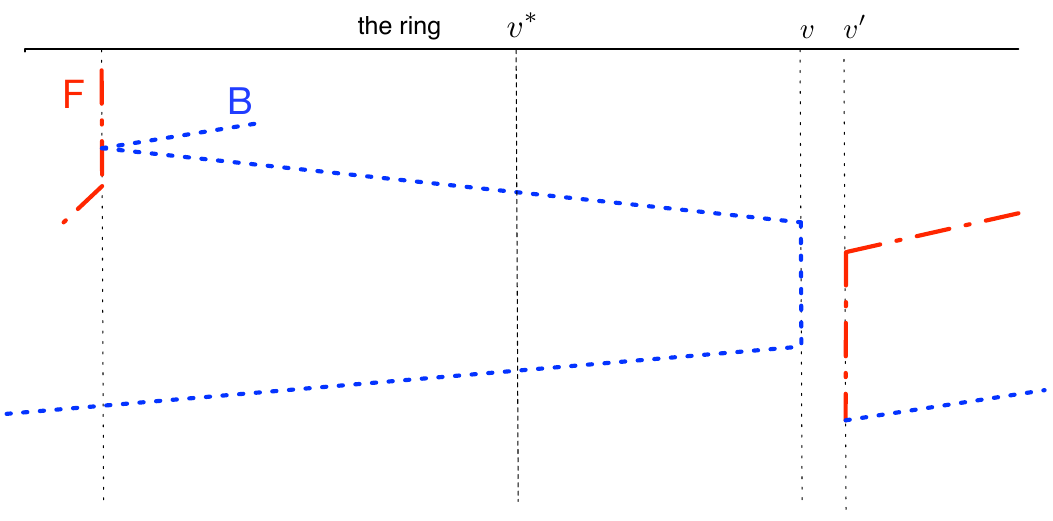}
\end{center}
\caption{ Agent $B$ signals $F$ to terminate, after checking that $\mathit{returnSteps} \leq 2 \cdot \mathit{bounceSteps}$. This implies that both agents were waiting on two endpoints of the same edge, in this case edge $(v,v')$.\label{figure:bsignaltermination}}

\end{figure}
\end{enumerate}

Consider now the cases when one agent terminates first (the case of agents never meeting is handled by Lemma~\ref{lm:noCatch}); we have two situations, in both we prove that the ring has been explored:

\begin{enumerate}
\item Agent $B$ terminates due to timeout $\vNtime > 3n$. As the agents are moving in the same direction, the number of successful moves differs by at most $n-1$. Since in each time step at least one of them advances, in less than $3n$ time steps from the moment when  $B$ learned $n$, $F $ will also learn $n$ and eventually terminate.

\item Agent $F$ terminates due to timeout $\vNtime\geq 7$.

%
%

Let $r$ be the round when  $F$ learned $n$.

If  $B$   entered state \sBounce at round $r' \leq r$,
 since the agents did not cross each other (satisfying the \pMeeting predicate),
agent $B$ switches to state \sReturn
at most $4n-2$ rounds from round $r$, as we now show. In fact, at round $r+4n-2$,  we have $\vEtime \geq 4n-2$.
If $\vEsteps <2n-1$, $B$ satisfies predicate ``$\vEtime > 2\vEsteps$" and thus enters state \sReturn. On the other hand, if $\vEsteps \geq 2n-1$, then $B$ has done at least a loop around the landmark $v^*$;
therefore ``$\vNtime > 0$" and thus $B$ enters   state \sReturn,
(see Figure \ref{fig:ssynclandmarkcase6} for an example).
Now the analogous argument as for  case 1 applies, both agents have the same direction and if $B$ catches $F$ it terminates. Therefore, if $B$ does not catch $F$, then it will
learn $n$ with at most $3n$ additional steps,  and eventually terminate.

Notice that, if  agent  $B$  enters state \sBounce  after round $r$, $B$ will catch  $F$ that  would signal  $B$ to terminate,
proving the lemma. Otherwise, if  $B$ does not catch $F$, that means that $B$ keeps staying in state \sReturn after round $r$, then it will
learn $n$ with at most $3n$ additional steps,  and eventually terminate.

\end{enumerate} 
\end{proof}

\begin{figure}
\begin{center}
\includegraphics[scale=0.49]{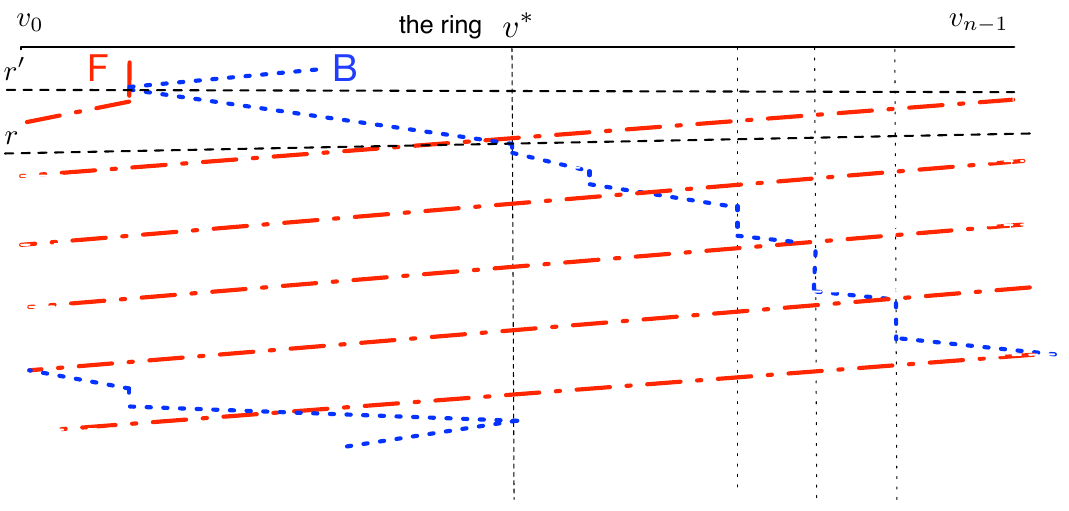}
\end{center}

\caption{Agent $F$ terminates due to timeout $\vNtime\geq 7n$: in this Figure we can see agent $F$ learning $n$ at round $r$, and agents $B$ switching state to \sReturn after $4n-2$ steps.\label{fig:ssynclandmarkcase6}}
\end{figure}

\begin{theorem}\label{thm:unknownNWithChirality}

Algorithm  {\sc LandmarkWithChirality}  allows  two anonymous agents
with  chirality to explore a 1-interval connected ring  with a landmark and to explicitly  terminate  in $O(n)$ time.
\end{theorem}
\begin{proof}
If the agents do not catch each other, the proof follows from Lemma~\ref{lm:noCatch}. Consider now the case that the agents catch each other at least once.
Also by Lemma~\ref{lm:noCatch}, we know that the meeting will happen no later than in round $7n-2$. The crucial observation is that, either the time between two consecutive meetings is linear in the progress made by  agent $F$, or the agents terminate following the catch.

\begin{figure*}
\footnotesize

\begin{framed}
\begin{algorithmic}
\State States: \{{\sf InitL}, {\sf Happy}, {\sf FirstBlockL}, {\sf AtLandmarkL}, {\sf Ready}, {\sf Reverse},  {\sf Bounce}, {\sf Return}, {\sf Forward}, {\sf Terminate}, {\sf BComm}, {\sf FComm}\}.
\AtState{InitL}{}
    \State $dir\gets\dLeft$, $k_1\gets 0$, $k_2\gets 0$, $k_3\gets 0$
    \State \Call{LExplore}{$dir \,|$ $n$ is known: \sHappy;  $\vBtime > 0$: \sFirstBlockL;  \pCatches: \sBounce; \pCatched: \sForward}
\AtState{Happy}
    \State \Call{LExplore}{$dir\,|$ $ \vTtime \geq 32((3\lceil \log(n) \rceil+3)5 \cdot n)+1$: \sTerminate;  \pCatches: \sBounce; \pCatched: \sForward}
\AtState{FirstBlockL}
    \State $dir\gets\dRight$, $k_1\gets \vTtime-1$
    \State \Call{LExplore}{$dir\,|$ $n$ is known: \sHappy, $isLandmark$: \sAtLandmarkL; $\vBtime> 0$: \sReady; \pCatches: \sBounce; \pCatched: \sForward}
\AtState{AtLandmarkL}
    \State $k_3\gets \vEtime$
    \If {both agents are at the landmark}
        \State Wait one round
        \If {both agents are at the landmark}
            \State \sTerminate
        \EndIf
    \EndIf
    \State \Call{LExplore}{$dir\,|$ $n$ is known: \sHappy, $\vBtime > 0$: \sReady; \pCatches: \sBounce; \pCatched: \sForward}
\AtState{Ready}
    \State $k_2\gets \vEtime$
    \State Compute your ID by interleaving bits of the bit-strings representation of $k_1$, $k_2$ and $k_3$. Each $k_i$ string of bits is padded by a prefix  $0$ until its length is equal to the biggest of the three.
    \State {\em set}(ID)
    \State Change to state \sReverse and process it
\AtState{Reverse}
    \State $dir\gets ${\em direction}$(\vTtime)$
    \If {$n$ is known}
        \State \Call{LExplore}{$dir\,|$ $\vTtime \geq  32((3\lceil \log(n) \rceil+3)5 \cdot n)$: \sTerminate; \pCatches: \sBounce; \pCatched: \sForward}
    \Else
        \State \Call{LExplore}{$dir\,|$ {\em switch}(\vTtime): \sReverse; \pCatches: \sBounce; \pCatched: \sForward}
    \EndIf
\AtState{Bounce, Return, Forward, BComm, FComm}
    \State The same as in Algorithm LandmarkWithChirality.
\END
\end{algorithmic}

\caption{Algorithm {\sc StartFromLandmarkNoChirality} \label{alg-StartFromLandmarkNoChirality}}
\end{framed}

\end{figure*}

Let $\mathit{pTime}_i$ denote the time between $i$-th and $i+1$-th catch and let  $\mathit{forwdSteps}_i$ be the progress made in that time  by agent $F$. We have:
$$\mathit{returnSteps}_i = \mathit{bounceSteps}_i + \mathit{forwdSteps}_i$$ Furthermore, $$\mathit{pTime}_i \leq 2 \cdot \mathit{bounceSteps}_i+1+\mathit{returnSteps}_i+\mathit{forwdSteps}_i$$
Where,  $2 \cdot \mathit{bounceSteps}_i+1$ is an upper bound on the time needed by agent $B$ to switch state from \sBounce to \sReturn given by the predicate ``$\vEtime > 2\vEsteps$".
The quantity $\mathit{returnSteps}_i+\mathit{forwdSteps}_i$ is an upper bound on the time needed for $B$ to catch again agent $F$.
Substituting $\mathit{returnSteps}$ into the latter yields $$\mathit{pTime}_i\leq 3 \cdot \mathit{bounceSteps}_i+2 \cdot \mathit{forwdSteps}_i+1$$ If the agents do not terminate after this catch, it must be $\mathit{forwdSteps}_i>\mathit{bounceSteps}_i$, hence $\mathit{pTime}_i\leq 6 \cdot \mathit{forwdSteps}_i$.
This means that  by   time $12n$ at the latest  since the first catch, agent $F$ will know $n$ and will terminate in $7n$ further rounds (if it does not terminate earlier due to some other terminating condition). The correctness now follows from Lemma~\ref{lm:termCorr}, and optimality is obvious. 
\end{proof}

\subsubsection {Termination:   Landmark without Chirality  \label{nochi} }


 In this subsection we  consider the case of a landmark when there is no chirality, and  we
 prove that  exploration  with termination can still  be performed   in  time $O(n \log n)$.

We first   consider and solve
 the problem when both agents start from the landmark;
  we then adapt the algorithm to work when agents start in arbitrary positions.

 \paragraph{Starting from the Landmark.}
 The pseudocode of Algorithm {\sc StartFromLandmarkNoChirality} is in Figure \ref{alg-StartFromLandmarkNoChirality}.
The main difficulty lies in the case when the agents start in   opposite directions and never break the symmetry.
Our approach to solve this case is to add an initial phase in which the agents use the event of waiting on a missing edge to break symmetry, obtain  different IDs (of size $O(\log n)$) and then use these IDs to ensure that, if the agents  do not catch each other (or outright explore the ring), then they eventually move in the same direction for a sufficiently long time  so that Algorithm {\sc  LandmarkWithChirality} succeeds.

Let us remark that, if the  agents somehow catch each other, they establish chirality, and then they can use Algorithm {\sc LandmarkWithChirality} which leads to exploration and termination. Therefore,  if at any point the agents catch each other, they enter states \sForward and \sBounce and proceed with Algorithm {\sc LandmarkWithChirality}.

 \underline{ Computing the ID:}  Each agent tries to compute its ID according to the  procedure described below. If an agent does not succeed in computing its ID,  then it has explored the ring and it is aware of that.

 \begin{figure}
\begin{center}
\includegraphics[scale=0.49]{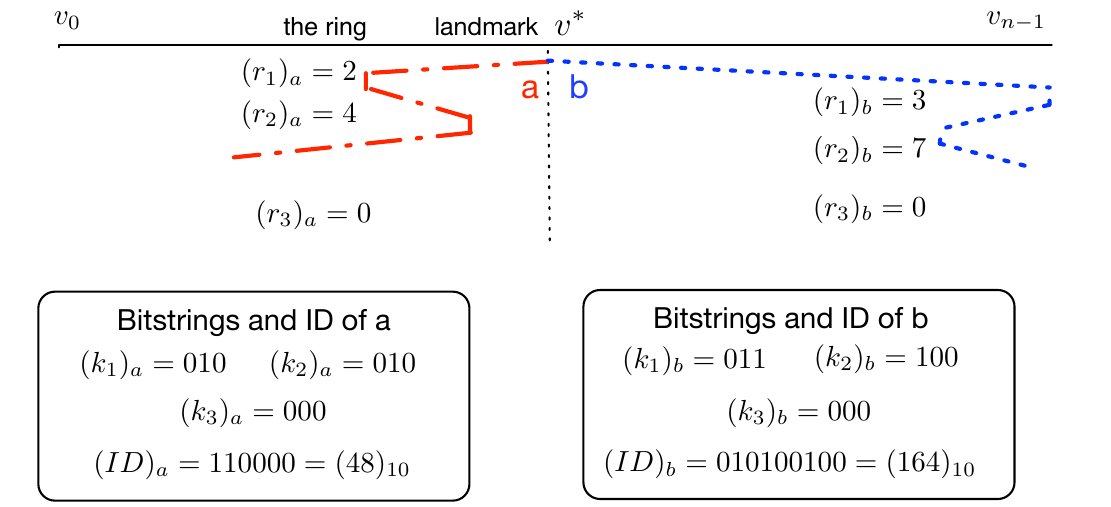}
\end{center}
\caption{Example of a run where different IDs are assigned to agents. In this case the round $r_3$ is $0$ for both agents; no one visits the landmark between $r_1$ and $r_2$.\label{fig:computIDexmpale}}

\end{figure}

  \begin{figure}
\begin{center}
\includegraphics[scale=0.49]{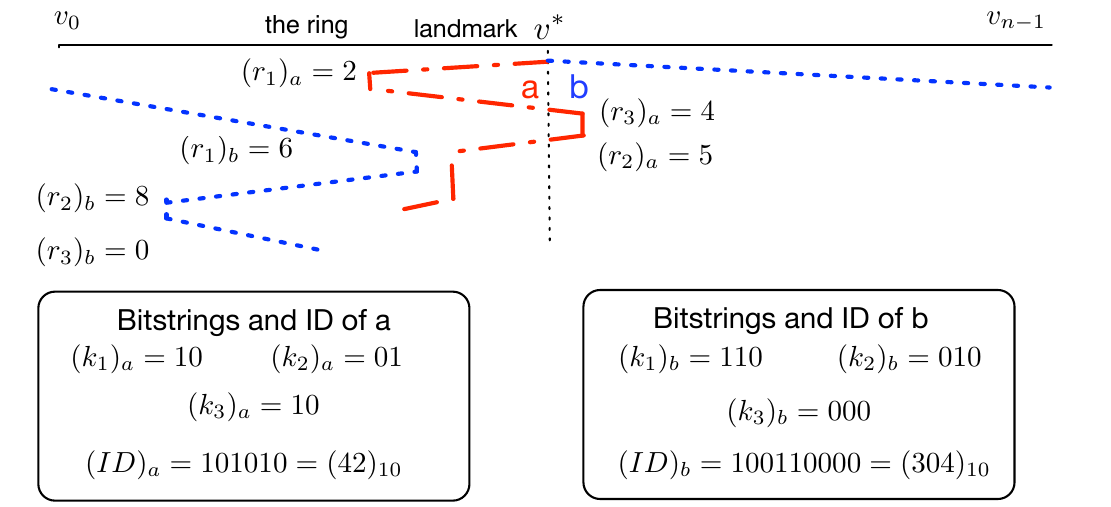}
\end{center}
\caption{Example of a run where different IDs are assigned to agents. In this case $r_3 \neq 0$ for agent $a$. \label{fig:computIDexmpale1}}

\end{figure}

 If  an agent does not know the ring size,  the first two times it waits in a port  it immediately changes direction.
 We indicate these rounds with $r_1$ and $r_2$, respectively.
 Let $r_3$ be the round when the agent entered the landmark for the first time between times $r_1$ and $r_2$ 
  ($r_3$ is set to $0$ if the agent does not traverse the landmark between rounds $r_1$ and $r_2$).
Let $k_2=(r_2-max(r_1,r_3))$  and  $k_3=max(0,(r_3-r_1))$;
 note that   
  these  values   are provided 
by variable  $\vEtime$ in the algorithm (see  Figure \ref{alg-StartFromLandmarkNoChirality}).
The computed ID of this agent consists of the interleaved bits of the numbers $k_1=r_1$, $k_2$ and $k_3$, where each bitstring is padded with a prefix of $0$'s until its length is equal to  that of the longest of the three.  Note that two IDs are equal if and only if their $k_i$'s are equal. In Figure \ref{fig:computIDexmpale} there is a detailed example on how the IDs are computed; notice that since both IDs start with $0$, this bit is ignored when the numerical value in base $10$ is obtained.  In Figure \ref{fig:computIDexmpale1} there is an example of  a run where $r_3 \neq 0$ for agent $a$.

Moreover, notice that if a round $r_1$ or $r_2$ does not exist, because the agent encountered a missing edge less than two times, then that agent has looped around the landmark; in this case it enters in the \sHappy State (cf. pseudocode). So it knows the ring size and it can compute an upper bound on the termination time of the other agent.

\underline{Using the IDs to decide the direction:}  The following procedure is used when an agent has computed its ID. Agents agree on a predetermined subdivision of rounds in phases.
Round  $r$ belongs to phase $j$, $r \in phase(j)$, iff $2^{j} \leq r < 2^{j+1}$. Given the ID, an agent computes a string of bits $S(ID)=10\circ(b($ID$))\circ0$,  where $\circ$ is the string concatenation and $b($ID$)$ is the minimal binary representation of ID.   Given a string $S$ we use $(S)_i$ to denote the $i-th$ character of $S$. Let us define as $\overline{j}$ the minimum value for which     $2^{\overline{j}} \geq len(S(ID))$ and $\overline{S(ID)}=(0)^{2^{\overline{j}}-len(S(ID))} \circ S(ID)$, where $len(S)$ is the length of the string $S$. For each phase $j \geq \overline{j}$ we associate the binary string  $d(ID,j)=Dup(\overline{S(ID)},2^{(j-\overline{j})})$, where $Dup(S,k)$ is the string obtained from $S$ by repeating each character $k$ times, e.g. $Dup(1010,2)=11001100$.

For each round $r \in phase(j)$, with $j > \overline{j}$, the direction of the agent is equal to $\mathit{left}$ if $(d(ID,j))_{r-2^j}=0$, otherwise it is $right$.
For a round $r \in phase(j)$ with $j \leq \overline{j}$ the direction of an agent is fixed to $\mathit{left}$.

In Figure \ref{fig:directions} there is an example of the bit sequences generated by an agent with $ID=1$.

In our algorithm this procedure is implemented using three functions:
\begin{itemize}
\item {\em set}$(ID)$: This function takes as parameter the ID of the agent, and it initializes the aforementioned procedure.
\item {\em direction}$(\vTtime)$: This function takes as parameter the current round and it returns the direction according to the aforementioned procedure.
\item {\em switch}$(\vTtime)$: This function takes as parameter the current round and  it returns true if $direction(\vTtime)$ $\neq direction(\vTtime-1)$.
\end{itemize}

\begin{figure}
\begin{center}
\includegraphics[scale=0.45]{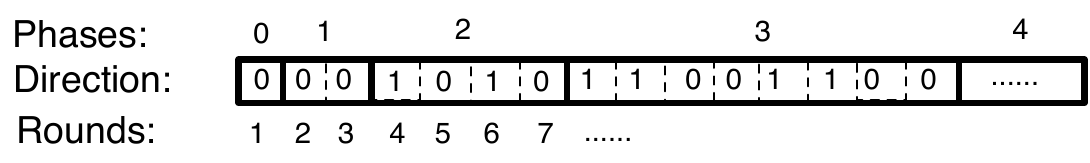}
\end{center}

\caption{ Directions for an agent with $ID=1$ \label{figure:nlogndirection}, a round with value $0/1$  corresponds to left/right direction.\label{fig:directions}}

\end{figure}

\begin{lemma}\label{lemma:direction}
Let us consider two agents with different IDs:\{$ID,ID'$\}, with $len(ID) \geq len(ID')$. For any constant $c >0$, by round $r < 32((len(ID)+3)c \cdot n)+1$ there has been a sequence of $c\cdot n$ rounds in which the agents had the same direction. Moreover,  by round $r$, each agent has moved in both directions for a sequence of rounds of length at least $c \cdot n$.
\end{lemma}
\begin{proof}
By definition $\overline{j}=\lceil \log(len(ID)+3)\rceil$. Phase $\overline{j}$ starts at round $r < 2^{\log(len(ID)+3)+2} < 4(len(ID)+3)$.

Consider $d(ID,\overline{j})$ and $d(ID',\overline{j})$. We now show that there exist two indices $x,y$ such that  $(d(ID,\overline{j}))_x$  $\neq (d(ID',\overline{j}))_x$ and $(d(ID,\overline{j}))_y = (d(ID',\overline{j}))_y$. It is easy to verify that $y=len(d(ID,\overline{j}))$, by construction since  $d(.)$ always terminate with $0$. For the index $x$ we consider two cases:
\begin{itemize}
\item $d(ID',\overline{j})=\overline{S(ID')}$: If  $len(S(ID))> len(S(ID'))$ then $x$ is the index of the first bit of $d_s(ID,\overline{j})$
different from zero. If $len(S(ID))=  len(S(ID'))$, since by assumption $ID \neq ID'$  index   $x$ exists.
\item $d(ID',\overline{j}) \neq \overline{S(ID')}$: By construction,   $d_s(ID',\overline{j})$ is composed by sequences of equal bits of length at least two. Also by construction, the first three bits of $S(ID)$ are  $101$. This means that the substring $101$ cannot be contained in $d(ID',\overline{j})$, and it is contained in $d(ID,\overline{j})$; this implies that index $x$ exists.
\end{itemize}
If the agents agree on the direction,   they will have the same direction in the round corresponding to index $y$, otherwise they will have the same direction in the round corresponding to index $x$.

In phase $j > \overline{j}$, by construction, we have a sequence  of rounds where agents have the same direction of length at least $2^{j -\overline{j}}$. We have $2^{j -\overline{j}} > c \cdot n$ when $j > \log(len(ID)+3)+\log(c\cdot n)+3$. We reach the end round of this phase by $r= \sum_{i=0}^{\log(len(ID)+3)+\log(c\cdot n)+4}2^i \leq 32((len(ID)+3)c \cdot n)$. The last statement of the lemma derives directly by the presence of $1$ and $0$ in each possible $S(ID)$. 

\end{proof}

\begin{theorem}\label{thm-StartFromLandmarkNoChirality}
Algorithm {\sc StartFromLandmarkNoChirality} allows  two anonymous agents without  chirality starting from the landmark
 to explore  a 1-interval connected ring with a landmark  and to explicitly terminate in $O(n\log(n))$ time.
\end{theorem}
\begin{proof}
First note that if the agents catch each other,  as shown in the proof of Theorem~\ref{thm:unknownNWithChirality}, they will explore the ring and terminate in $O(n)$ time after the moment they catch; hence, in the remainder of the proof,  we deal with the case when  the agents never catch each other.
\begin{figure}
\begin{center}
\includegraphics[scale=0.5]{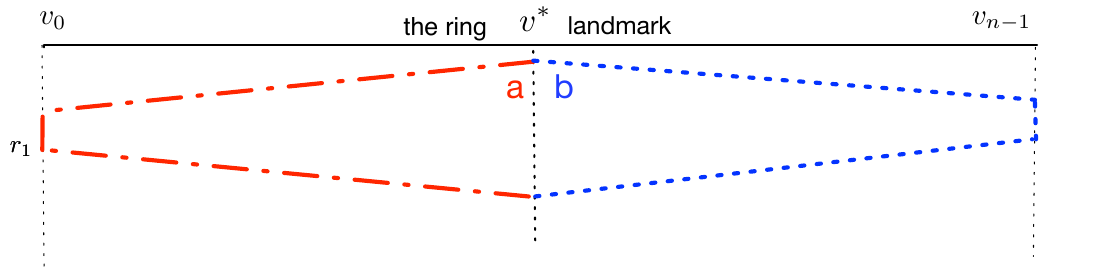}
\end{center}
\caption{Termination from state \sAtLandmark.\label{fig:termatlandmark}}

\end{figure}
Next note that, if the agents meet at the landmark and terminate in state \sAtLandmark, they must have bounced from the same edge and the ring has been explored; this is because they started from the landmark and returned at the same time while both were blocked exactly once; see Figure \ref{fig:termatlandmark}.
However, when two agents meet in the landmark and one terminates,  to ensure that the other is in state \sAtLandmark
the following synchronization step is needed: when an agent enters   state \sAtLandmark and sees also the other agent, it waits one round more in the node without moving.
If the other agent does the same they both terminate.
This obviously happens in the aforementioned case, i.e. when the agents bounced on the same edge and reach the landmark at the same time, entering both in the same round in state \sAtLandmark, thus correctly identifying the exploration of the ring.
The same cannot happen if one of the agents is not in the \sAtLandmark state; in this case only one of the agent will wait and the other will either leave the node, or it will enter a port, thus preventing a possible incorrect termination.

Third, observe that, by   time $3n-1$, either an agent knows $n$ (and terminates in $O(n\log(n))$ time from \sHappy state, or it knows its own ID. Note that IDs are bounded from above by $n^3$, since each $k_i$ is at most $n$, which implies $len(ID) \leq 3 \lceil \log(n) \rceil$.

Consider now the case that at time $3n-1$ an  agent (say $a$) does not know its ID (and hence since time $3n-1$ knows $n$), while the other ($b$) knows its ID but does not know $n$.  Agent $b$ therefore repeatedly switches its direction in state \sReverse, while   agent $a$ moves in the same direction. Note  that by  Lemma \ref{lemma:direction}, by time $32((3 \lceil \log(n) \rceil +3)5 \cdot n)+1$,  agent $b$ has moved to the $\mathit{left}$ and $right$ direction for a sequence of rounds of length at least $5n$, in one of the two both $a$ and $b$ move in the same direction. As at least  one agent makes progress in each of those time steps, while (by assumption) they don't catch each other, $b$ must have moved for at least $2n$ time units. This means that $b$  learns $n$ and eventually terminates as well.

The final case to consider is when both agents know their IDs, but do not    know  $n$. Note that if the agents have the same values of $k_1$ and $k_2$, they must have covered the whole ring and at least one of them will have $k_3\neq 0$. This means that the agents necessarily have different IDs, since if they had the same values of $k_1$ and $k_3\neq 0$, they would have terminated in \sAtLandmark state, see Figure \ref{fig:termatlandmark}.

Since the IDs are different, by Lemma \ref{lemma:direction}, by round $32((3\lceil \log(n) \rceil+3)5 \cdot n)+1$ there has been a time segment of length $5n$ in which both agents were moving in the same direction. Thus, either they catch each other, or both learnt $n$ and terminated thereafter. 
\end{proof}

\begin{figure*}
\footnotesize

\begin{framed}
\begin{algorithmic}
\State States: \{{\sf Init}, {\sf AtLandmark}, {\sf InitL}, {\sf FirstBlock}, {\sf FirstBlockL}, {\sf AtLandmarkL}, {\sf Ready}, {\sf Reverse},  {\sf Bounce}, {\sf Return}, {\sf Forward}, {\sf Terminate}, {\sf BComm}, {\sf FComm}\}.
\AtState{Init}{}
    \State $dir\gets\dLeft$, $k_1\gets 0$, $k_2\gets 0$, $k_3\gets 0$
    \State \Call{LExplore}{$dir \,|$ $n$ is known: \sHappy; $\vBtime > 0$: \sFirstBlock;  \pCatches: \sBounce; \pCatched: \sForward;}
\AtState{FirstBlock}
    \State $dir\gets\dRight$, $k_1\gets \vTtime$
    \State \Call{LExplore}{$dir\,|$ $n$ is known: \sHappy; $isLandmark$: \sAtLandmark; $\vBtime > 0$: \sReady; \pCatches: \sBounce; \pCatched: \sForward}
\AtState{AtLandmark}
        \State $k_3\gets \vEtime$
    \If {both agents are at the landmark}
        \State Wait one round
        \If {both agents are at the landmark}
            \State go to state {\sf InitL}
        \EndIf
        \EndIf
        \State \Call{LExplore}{$dir\,|$ $n$ is known: \sHappy; $\vBtime > 0$: \sReady; \pCatches: \sBounce; \pCatched: \sForward}

\AtState{$S \not\in \{$Init, FirstBlock, AtLandmark$\}$}
    \State The same as in Algorithm StartAtLandmarkNoChirality.
\END
\end{algorithmic}
\caption{Algorithm {\sc LandmarkNoChirality} \label{alg-LandmarkNoChirality}}
\end{framed}
\end{figure*}

\paragraph{Arbitrary initial  positions. }

Algorithm {\sc StartAtLandmarkNoChirality} almost works also in the case of agents starting in arbitrary  position.
The only failure would be due to the fact that,
 when the agents meet in the landmark while  establishing
 $k_1$ and $k_2$, it does not necessarily mean that they have already explored the ring.
The modification to introduce  is not to terminate in this case, but to reset and start a new instance in state \sInitL, executing algorithm {\sc StartAtLandmarkNoChirality}, as now the agents are indeed starting at the landmark.
If the agents do not meet at the landmark, then their values of $k_3$ are different and the algorithm works using the same arguments.
The complete pseudocode is in Figure \ref{alg-LandmarkNoChirality}. Since this adds at most $O(n)$ to the overall time, we obtain the following theorem.

\begin{theorem}\label{thm:LandmarkNoChilarity}
Algorithm {\sc  LandmarkNoChirality}  allows  two anonymous agents 
without  chirality to explore  a 1-interval connected ring with a landmark and to explicitly terminate in $O(n\log(n))$ time.
\end{theorem}

\section{Ring exploration in ${\cal SSYNC}$ \label{ssync}}

In this section we investigate the  exploration problem  when the system is {\em semi-synchronous}. The complexity measure we consider in this case is the total number of edges traversed by the agents.
As in Section \ref{sync},   \vEsteps denotes the total number of  successful moves performed by the agent since   procedure \Explore has been called, \vTnodes denote the total number of nodes that the agent perceived to have explored since the beginning of the protocol,
 \vBtime denotes the number of consecutive rounds the agent has been currently waiting in a port,
 and \pCatches is the predicate denoting that the agent is in the node and  the other agent is observed on a port (in the moving direction of the first agent).
 These definitions hold for all the algorithms of the ${\cal SSYNC}$ section.

\subsection{Impossibility of Exploration in {\sf NS}}
Let us begin by showing an intuitive result for the {\sf NS} model:


\begin{theorem}\label{th:impNS}
In the {\sf NS} model, exploring the ring is impossible with any number of agents,  even if the ring and the agents are not anonymous and there is
chirality.
\label{nsimpossible}
\end{theorem}
\begin{proof}
Consider a non-anonymous ring  where $k >1$ non-anonymous agents are located at  arbitrary nodes, with at least a node without agents; and 
let   ${\cal A}$ be an arbitrary exploration algorithm.  Starting with $t=0$,
let  $A(t)$ denote the set of agents that, according to ${\cal A}$, if active at time $t$  
would want to move to 
a neighbouring node;  
let $P(t)$  denote the set of agents that instead would not move;
and let $first(t)\in A(t)$ be the agent in $A(t)$ that has
not been active the longest, where ties are arbitrarily broken.

Consider now the following  agents and link activation scheduler:  at time $t$ it activates only $P(t)$ and $first(t)$, and 
it removes the  edge on which $first(t)$ would move. Hence, no agent will move at time $t$. 
By repeating this process and by observing that this  scheduler is indeed fair (it activates every  agent infinitey often),
the nodes initially without an agent will never be visited.
\end{proof}

 Notice that Theorem \ref{th:impNS} would hold  even if the agents were equipped with  wireless communication.

%

Motivated by this impossibility result, we now examine the other ${\cal SSYNC}$ models.

\subsection{Exploration in {\sf PT}}


\subsubsection{{\sf PT}: Impossibility of Exploration by Two Agents Without Chirality}
We begin our investigation of the {\sf PT} Model by showing that, without chirality, two agents cannot explore the ring, even with precise knowledge of the network size and with the presence of a landmark.

\begin{theorem}\label{twonotenoughpt}
In the {\sf PT} model without chirality two anonymous agents are not sufficient  to explore a ring of size $n \geq 5$. The result holds even if there is a distinguished landmark node and the exact network size is known to the agents.
\end{theorem}
\begin{proof}

By contradiction, let ${\cal A}$ be a solution algorithm. Let $a$ and $b$ be the two agents.
Assume that:  agent $a$ starts at  node $u$ and, according to ${\cal A}$, it  would move  towards $u'$;
     $b$ starts at node $v$ and would  move towards $v'$; and $u$, $u'$, $v$, $v'$ and the landmark are all different.

The algorithm is executed by the agents against an adversary that choses which agent is active in which round,
decides the local orientation of the agents, as well as the topological structure of the ring.
In particular, it does not fix the link relationship between the nodes where $a$ acts ($u, u'$) and  those
where $b$ operates ($v, v'$), until necessary during the execution of the algorithm ${\cal A}$.

The adversary  applies the following {\em alternation strategy}:
It activates only $a$ and keeps it active until
$a$ tries to move to  a node other than    $u$ and  $u'$.
 At this point, the adversary  blocks that edge and  keeps it blocked until $a$ either switches direction, or decides to permanently wait on this port at node  $u^*\in\{u,u'\}$. Notice that one of these two events has to take place.
  When this happens,   the adversary makes $a$ passive,  activates $b$,
  and acts with $b$  exactly in the same way as it did with agent $a$. In other words,
  it keeps $b$ active until  $b$  tries to move  to  a node other than    $v$ and  $v'$ and then
it   blocks that edge and   keeps it blocked until $b$ either switches direction, or decides to permanently wait on this port at node  $v^*\in\{v,v'\}$.

Note that, after the execution of the alternation strategy, since the agents are anonymous and there is no chirality, whatever decision $a$ made (switch or permanently wait), the same is  taken by agent $b$.

Let us consider first  the case when   both agents decide to switch direction.
 In this case, the adversary will continue,  as long as the two agents decide to switch direction, to execute the alternating strategy, never letting the agents move outside the four nodes $u$, $u'$, $v$, $v'$. This means that, if they never decide to wait permanently,  the rest of the ring will not be explored, contradicting the correctness of ${\cal A}$.

Therefore, within finite time, after some executions of the alternating strategy,
  they both must decide to wait permanently, $a$ at node  $u^*$
to go to node   $u''$, and $b$ at node  $v^*$
to go to node   $v''$.
 When this happens,  the adversary  fixes the topology of the ring  by setting
$u'' = v^*$ and $v''=u^*$; it then  blocks forever  edge $(u^*,v^*)$,
and permanently activates both $a$ and $b$.
Thus the rest of the ring will not be explored, contradicting the correctness of ${\cal A}$. 

\end{proof}

As a consequence of Theorem~\ref{twonotenoughpt},  any exploration algorithm must either use chirality, or   employ more than two agents.
We will consider the two cases in the following subsections.


%

\subsubsection{{\sf PT}:  Two Agents with Chirality} \label{ptchiral2agents}

We know  (Theorem \ref{twonotenoughpt}) that two agents need chirality to explore the ring; however  chirality does not suffice
for exploration with termination  (Theorem \ref{impossibilitywithIDterminationsync}) unless the agents know an
upper bound on the ring size or  the nodes are not anonymous (i.e., there is a landmark).

Interestingly, in the  {\sf PT} model, even with both knowledge of the the ring size and  a landmark,
explicit termination of both agents is impossible
 as shown by the following theorem.

\begin{theorem}
In the {\sf PT} model,   any exploration algorithm on a ring of size $n \geq 3$ with two agents can only guarantee partial termination.
The result holds even if the size of the ring is known, a landmark node is present and there is chirality.
 \label{parttermination}
\end{theorem}
\begin{proof}
%
%

By contradiction, let ${\cal A}$  be a  two-agents exploration algorithm   for  {\sf PT}, where both
agents always terminate in every execution.
Let $a,b$ be the two agents, and let them start at the same location, the landmark node $v_s$, both initially asleep.
First notice that, if only one agent wakes up and the other remains asleep,
eventually the awake agent has to start exploring the ring (otherwise, by alternating the sleeping of $a$ and $b$, they will never move).

Let $E(a)$ (resp. $E(b)$) be the unfair execution of ${\cal A}$ in which $a$ (resp. $b$) forever sleeps,  
$b$ (resp. $a$) is active and explores the ring, 
and no edge is removed.
Observe that,  since  the agents are anonymous, 
 then $a$ does in $E(b)$ the same moves as $b$ does in $E(a)$. 

Let us now examine the four  possible situations in which $b$ can find itself in $E(a)$:

%

\begin{enumerate}

\item  

	Agent $b$ terminates in some location $v_x\neq v_s$. Consider now the execution  
	$E'$ coincident with $E(a)$ until the round when $b$ terminates in  $v_x$; then
	 the adversary wakes up $a$ and  blocks it on  $v_s$; hence $a$ will never leave $v_s$. 
	 Observe that $E'$ is a fair execution;
	 since, by assumption, both agents  explicitly terminate in every fair execution,
	then    $a$  will terminate in $E'$ on node $v_s$, say  $t$ rounds after becoming awake.
	
	Consider now an execution $E''$ coincident with $E(a)$ until $b$ moves to a node $v_{y}\neq v_s$;
	when  this happens
	 $a$ is woken up and blocked
	 on node $v_s$, while $b$ is kept asleep.	
	From the local view of $a$, the executions $E'$ and $E''$ are not distinguishable; hence 
	$a$  terminates in $E''$  on $v_s$  $t$ rounds after waking up. When this happens, the adversary wakes up $b$ and blocks
	it on node $v_y$. Observe that also $E''$ is a fair execution.
	As $n \geq 3$ and only $v_s$ and $v_y$ are visited in this execution, the ring will never be explored, contradicting the correctness of ${\cal A}$.

\item  
Agent $b$ terminates on node $v_s$.  
Consider now the execution  
	$E'$ coincident with $E(a)$ until the round when $b$ terminates in  $v_s$;
when this happens the adversary wakes up $a$.
Since $a$ does not know that $b$ has terminated and  $b$'s behaviour is indistinguishable
from being asleep, $a$  has to leave $v_s$. When it  reaches a neighbouring node $v_y$,
the adversary blocks  it there. Notice that $E'$ is a fair execution.
Consider next the execution  
$E''$ coincident with $E(b)$ (i.e., $a$ is made active while $b$ is asleep) until 
$a$  leaves $v_s$; when it  reaches  a neighbouring node $v_y$, also in this case
the adversary blocks  it there.

Observe that  $a$ cannot distinguish between executions $E'$ and $E''$; hence $a$ will take the same decisions in both.
 If it terminates within finite time, then, when that happens, in  execution $E''$ the adversary awakes $b$
 and  perpetually blocks  it  on  $v_s$;
 hence the ring will never be explored. Since $E''$ is a fair execution, this contradicts the correctness of ${\cal A}$.
 If, on the other hand, $a$ does not terminate,   execution $E'$   contradicts the assumption that  ${\cal A}$
always guarantees explicit termination.

\item 

Agent $b$ does not terminate, and it visits $v_s$ only finitely many times. Hence, there is a round $r$ after which $b$ does not visit $v_s$ 
($b$ can be waiting at some node $v_x\neq v_s$; it can also be perpetually moving among a set of nodes not containing $v_s$).

Consider now the execution  
	$E'$ coincident with $E(a)$ for the first $r$ rounds;  the adversary then  wakes  $a$ at round $r+1$ and blocks it on $v_s$ forever; notice that $E'$ is fair.
 By assumption, both agents terminate in ${\cal A}$. In particular, this means that $a$ terminates on $E'$ without leaving $v_s$. 
 
 Consider now an execution $E''$ coincident with $E'$   until the first time $b$ visits a node $v_y\neq v_s$. 
 From that moment on, $b$  becomes asleep in $v_y$, while $a$ is activated and blocked on $v_s$. 
 From the point of view of $a$ this execution is undistinguishable from $E'$, therefore $a$ terminates in $v_s$. After the termination of $a$, 
 the adversary wakes up   $b$  in $E''$ but  perpetually blocks it; notice that $E''$ is fair.
 As only $v_s$ and $v_y$ have been explored in this execution, ${\cal A}$ fails to explore the ring.

\item 
Agent $b$ does not terminate, and it visits $v_s$ infinitely often, either by reaching it and no longer moving  or by moving inside an interval  $I$ of nodes containing $v_s$. 
Consider now the execution $E'$ coincident with $E(a)$ until $b$ 
returns to  $v_s$; the adversary then puts $b$ to sleep 
and  wakes up $a$ until it 
returns to  $v_s$. 

Now, the adversary will  repeate the whole process,
alternating forever the sleep of $a$ and $b$:
wake up $b$ (resp. $a$) and put $a$ (resp. $b$)  to sleep, util the next round in which $b$ (resp. $a$) is at $v_s$ (recall that  it might never leave it).
Notice that  $b$ is unable  to distinguish this fair execution $E$ from  $E(a)$, and similarly $a$ is unable  to distinguish it from  $E(b)$;
hence they will never terminate in this execution.

\end{enumerate}
\end{proof}

\noindent As a consequence, the best that can be achieved is partial termination.

In the rest of this section, we show that the knowledge of  an upper bound or the presence of a landmark is sufficient
for the two agents to explore the ring;  with respect to termination, we achieve a strong partial termination,
 with {one} agent always explicitly terminating, and the other  either terminating or  no longer moving.

\paragraph{A. Chirality and Known Upper Bound}
\ \\
First consider the case when an upper bound $N$ on the ring size is known to the agents.
We present an algorithm, {\sc {\sf PT}BoundWithChirality}, shown in Figure \ref{alg-ptbchirality}, for exploration of dynamic ring by two agents with chirality and knowledge of an upper bound.
Both agents start moving $left$. If an agent
 finds a blocked edge with the other agent waiting in the left port,
 it  enters state   {\sf Bounce}, reverses direction  and starts moving $right$.
If the agent in   state   {\sf Bounce}
 finds a missing edge before traversing $N$ edges,  it reverses direction again  (becoming {\sf Reverse});   that agent might   be alternating   {\sf Bounce}  and  {\sf Reverse} state several times.
 An agent terminates    upon discovering  it has traversed $N$  consecutive edges in a given direction.  Additionally,  if  in  state {\sf Reverse}   the  agent
  catches  the other agent
   at a distance smaller than that in the previous catch,
   it means the two agents have crossed and they can safely terminate.

\begin{figure}
\footnotesize
\begin{framed}
\begin{algorithmic}

\State States: \{{\sf Init}, {\sf Bounce}, {\sf Reverse}\}.
\State $\mathit{leftSteps} \gets \bot$
\State $\mathit{rightSteps} \gets \bot$

\AtState{Init}{}
    \State \Call{Explore}{\dLeft$|$ \vTnodes $\geq N$: \sTerminate,  \pCatches: \sBounce}
\AtState{Bounce}{}
	\State $\mathit{leftSteps} \gets$\vEsteps
	\If{ $(\mathit{rightSteps} \neq \bot) \land  (\mathit{rightSteps} \geq \mathit{leftSteps})$ }
		\State {\sf Terminate}
	\EndIf
    \State \Call{Explore}{\dRight$|$ \vTnodes $\geq N$: \sTerminate, \vBtime $> 0$: \sReverse }
\AtState{Reverse}{}
	\State $\mathit{rightSteps} \gets$\vEsteps
    \State \Call{Explore}{\dLeft$|$ \vTnodes $\geq N$: \sTerminate, \pCatches: \sBounce}

\END
\end{algorithmic}
\caption{Algorithm {\sc {\sf PT}BoundWithChirality} \label{alg-ptbchirality}}

\end{framed}
\end{figure}

\begin{theorem}\label{PTupperbound}
Two agents executing Algorithm\\ {\sc {\sf PT}BoundWithChirality} in the {\sf PT} model with   chirality  and a known upper bound $N$ on the size of the ring will explore the ring using at most $O(N^2)$ edge traversals. Furthermore, one agent explicitly terminates, while the other  either terminates or it waits perpetually on a port.
\end{theorem}

\begin{proof}
We first  prove exploration. Note that, by definition of $\vTnodes$,  if it exceeds $N$, the ring has been explored. The only non-trivial case to consider is   termination due to the condition: $\mathit{rightSteps} \geq \mathit{leftSteps}$.

Let us first consider the case when where agent $b$ terminates
after it bounces on agent $a$ blocked on edge $e$ at round $r$,  then changes direction on a missing edge and finally catches $a$ again terminating.

If $a$ stayed in the same port during all this process,   then $b$ would have visited  the other side of   $e$
 (otherwise the  {\sf PT} condition would have ensured passive transport of $a$). In such a case  the ring has been explored.

If $a$ moved after round $r$, two different scenarios are possible:
\begin{itemize}
\item $(i)$ $a$ crosses $b$  while $b$ is in \sBounce state, this implies that $a$ and $b$ explored the ring.

\item  $(ii)$  $a$ bounced on $b$ while $b$ was in \sReverse, $a$ traverses $e$ in state \sBounce, and $a$ goes back in state \sReverse. It is clear that also in this case the ring has been explored: the portion between $e$ and the node where $a$ bounced on $b$ has been explored by $b$, and the other portion by $a$.
\end{itemize}


\begin{figure}
\begin{center}
\includegraphics[scale=0.52]{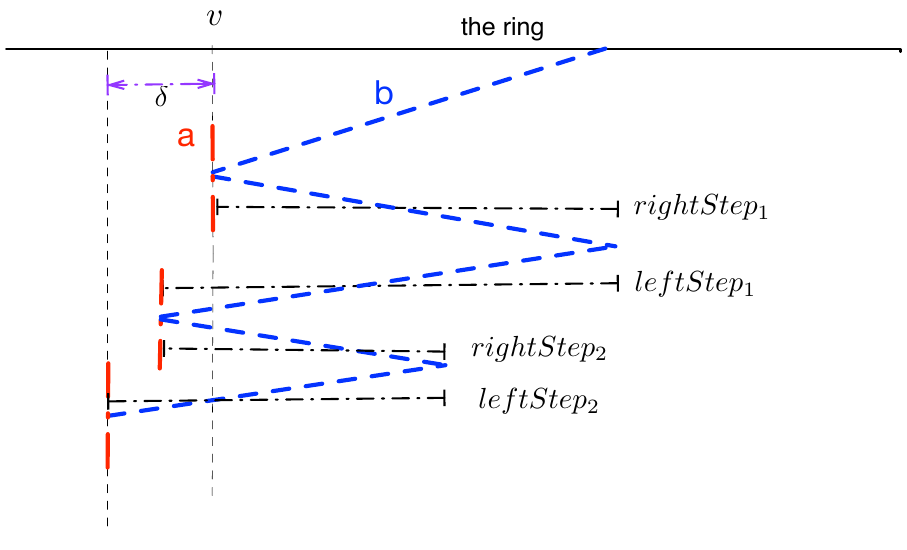}
\end{center}
\caption{Run where  $\mathit{rightSteps} < \mathit{leftSteps}$, the $\delta$ grows at each bounce-reverse of $b$. \label{fig:zigzagincreasing}}
\end{figure}

We now  have to show that at least one agent terminates. If the adversary keeps an edge perpetually removed,  eventually the algorithm terminates due to condition $\mathit{rightSteps} = \mathit{leftSteps}$. Moreover, if an agent is not blocked in its traversal, it will eventually do $N$ steps leading to termination.

 The only possibility left  to be analyzed is if $a$ is blocked on some edge $e$, $b$ bounces first on edge $e$,  then reverses direction on edge $e'$ and, when $b$ catches on $a$, $\mathit{rightSteps} < \mathit{leftSteps}$. Notice that,  when this happens, $b$ has done at least one step further to the left of edge $e$, otherwise we would have $\mathit{rightSteps} \geq \mathit{leftSteps}$ leading to the termination of $b$.

  We have that (1) the adversary cannot keep an agent sleeping, or blocked, forever, (2) the adversary cannot let an agent do $N$ steps in one direction and (3) an agent switches from state \sReverse to \sBounce only catching someone; therefore if the agents do not terminate, then $a$ and $b$ have to catch each other multiple times. As  discussed previously, for each catch the catching agent will do an additional step   to the left.  Therefore, we will eventually have $\vTnodes > N$ for one of the two.
An example is reported in Figure \ref{fig:zigzagincreasing}; in the figure we can see that after the first bounce on node $v$, the area to the left of $v$ explored by $b$  (named $\delta$ in the figure)  has to grow at least by one for each sequence of \sReverse-\sBounce: we must have $\mathit{rightSteps}_{i} < leftStep_{i}$.

  If an agent terminates, the other   cannot bounce to the right. Hence, it will either terminate due to exceeding $N$ left moves, or will be perpetually blocked on a port, and the last part of the theorem holds.

Let us now analyze the complexity of the algorithm. Observe that during one \sBounce-\sReverse phase an agent can do $O(N)$ steps. There could be at most $N$ of these \sBounce-\sReverse phases: in each of them the agent has to do an additional step left otherwise the termination condition is satisfied. Since the termination check bounds the total number of left steps by $N$, this yields $O(N^{2})$   complexity of the algorithm. 
\end{proof}

The complexity of the proposed algorithm, $O(N^{2})$, depends on the accuracy
of the upper bound $N$. We will now show that such a dependency is to a certain degree
 inevitable. In fact we prove  that $\Omega(N \cdot n)$ are indeed required by any solution protocol.
 This also means that the proposed algorithm is {\em optimal}, whenever $N=O(n)$.

\begin{theorem}\label{th:lowerboundmovespt}
In the {\sf PT} model with chirality in which the two agents know an upper bound $N$ on the ring size $n \geq 5$,
 any   exploration algorithm with partial termination requires
  at least $\Omega(N \cdot n)$ edge traversals
by the agents in the worst case.
\end{theorem}
\begin{proof}

Let ${\cal A}$ be a solution algorithm for the   {\sf PT} model with chirality and let $N$ be the known upper bound on the ring size.
The actual size $n\leq N$ of the ring will be decided by the adversary.
The adversary will operate in logical phases; in  phase $i$, it selects
 a continuous  segment $X_i$ of the ring that includes both agents;
 initially   $X_1$  has size $x \geq 5$.
In each phase, the adversary
 uses the  strategy described below.
Let $N_i$ denote the number of distinct nodes that the agents would have explored,
under this strategy,
at the beginning of phase $i$ if   $n=N$, i.e. the upper bound was tight.
We assume, without  loss of generality,   that
$N_1 = |X_1| = x$; this only makes the lower bound proof stronger.
The adversary repeats the same process as long as
 $N_i < N$.

 In  phase $i$
the adversary considers one agent at a time.
 If that agent, say $a$,
  tries to leave $X_i$, the adversary will block that edge
  and keeps it blocked until either $a$ changes direction
  (i.e., stays  within $X_i$)
  or also the other agent tries to leave $X_i$ (from opposite directions).

 Eventually both agents must  want to leave $X_i$ from
opposite directions,
otherwise the adversary can continue to prevent the agents
from leaving $X_i$ and
since $N_i <  N$, choosing $n>N_i$ would make the algorithm ${\cal A}$  fail.

When the adversary determines that also the second agent, say $b$,
would want to  leave $X_i$ should it become active (recall  the first agent is  currently
waiting on a blocked edge),
the adversary
 considers what  each agent would do, according to ${\cal A}$, if  the edge
 the agent  wants to cross is  blocked indefinitely.

Clearly,  they cannot both wait indefinitely, otherwise  the adversary,
by choosing $n=N_i$ and activating both agents at all times,
would make the two agents wait forever at the two endpoints of the same blocked edge, without either of them terminating (although they have explored the ring).

This means that  at least one of the agents,
after  a finite number of activations, would  change direction and
move towards a neighbouring node in  $X_i$.
Let $c\in\{a,b\}$ be such an agent, and $d$ be the other.
Upon this determination,
if $d=b$,   the adversary  activates $d$ (letting it move),
 and then makes $d$ passive; otherwise,  i.e. $d=a$,  it makes $d$ passive immediately, and let it move  passively on the next node.
In any case, it also  blocks the edge from which
$c$ is trying to leave $X_i$
until $c$  changes direction and moves from its current position
 $u$ to the next node $u'$ in $X_i$.
At this point, the adversary blocks the edge  $(u,u')$ and keeps it blocked until $c$
performs at least $|X_i|-1$ additional moves, possibly reaching $d$. Note that this will always happen as shown by the claim below.

\noindent {\bf Claim 1.}
{\em   If $c$ never reaches $d$ in this phase and it does not wait forever in the port of node $u'$ trying to reach node $u$, it will perform
an unbounded number of moves.}
 \begin{proof}

By contradiction, let $c$ satisfy the conditions of the claim, but perform only a bounded number of moves, without reaching $d$.
Thus, within finite time it will stop in a node indefinitely.
 When this happens, the adversary activates $d$ blocking any edge through which
it wants to move. Hence no new node will be explored; since
 $N_i <  N$,
by choosing $n=N$ the adversary would make the algorithm ${\cal A}$  fail.

\end{proof}

Further note that, during this time, agent $d$ moves (if $d=b$) or is passively
moved (if $d=a$) to a neighbour $v$  outside $X_i$.

  In   {\bf Claim 1} we   assumed  that  $c$ does not wait forever in the port of node $u'$ trying to reach node $u$.
Notice that if this happens, then we can define the set $X_{i+1}=(X_i \setminus \{u\} )\cup \{v\}$.

If $d$ is blocked traversing an edge to go outside $X_{i+1}$,  it cannot wait forever, otherwise the adversary can create a schedule and
a configuration in which both agents are forever blocked on the endpoints of edge $(u,u')$, without the possibility of termination.

An example is shown in Figure \ref{fig:imp1}: in this run we have that agent $c$ is blocked until round $r_0$, when it decides  to  change direction. When  $c$  reaches $u'$ at round $r_1$, it is forced to sleep and agent $d$ is allowed to move reaching $u$. When this happens, agents $c$ is activated again and it   tries to traverse edge $(u,u')$ at round $r_2$. At round $r_3$, also agent $d$   tries to traverse edge  $(u,u')$, and both agents keep doing that forever without terminating. Notice, that the agents only know an upper bound $N$ on the ring size.
\begin{figure}
\begin{center}
\includegraphics[scale=0.43]{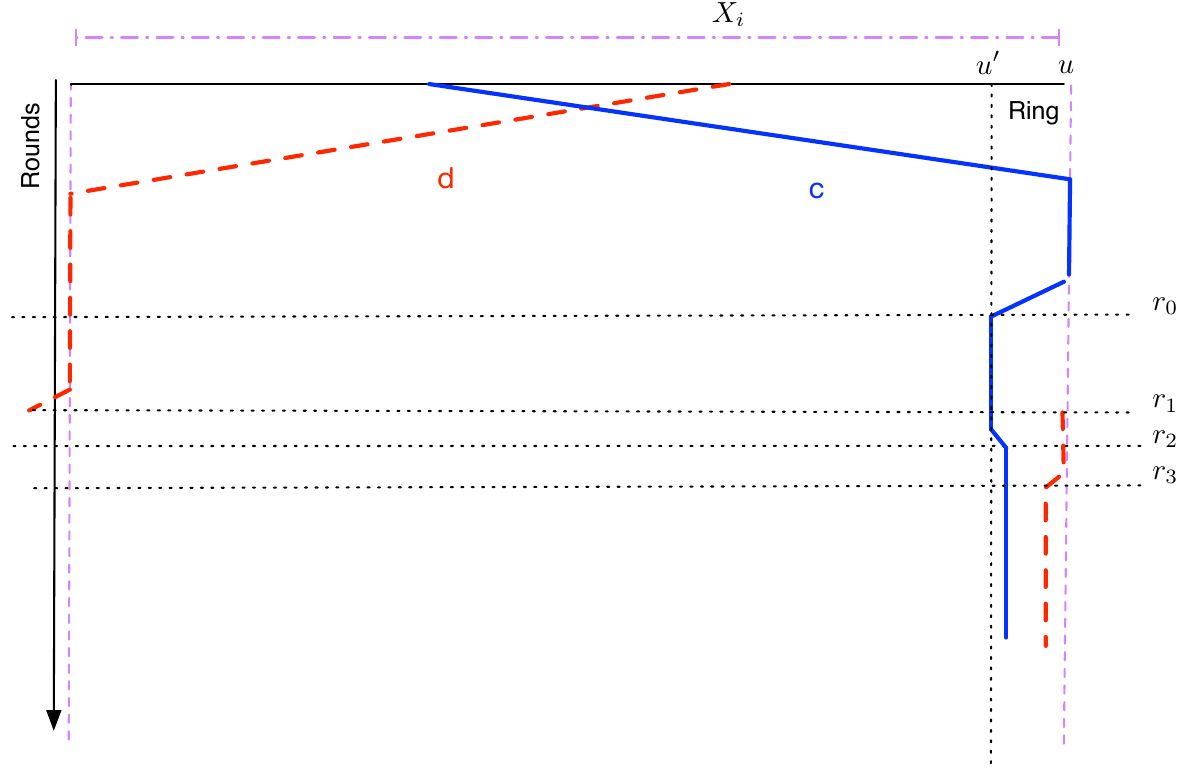}
\end{center}
\caption{Run where agents end up waiting on the same edge. \label{fig:imp1}}
\end{figure}

Therefore, agent $d$ eventually leaves node $v$ to go to the neighbouring node in $X_{i+1}$, thus the adversary can let $c$ visit node $u$, then it forces $c$ to sleep and it keeps $d$ blocked until $d$ leaves $v$.
At this point we are again in the situation where both agents are confined in the set of nodes $X_i=(X_{i+1} \setminus \{v\} )\cup \{u\}$.
If the agents keep behaving in this way,  the adversary makes them oscillating forever between the sets $X_i$ and $X_{i+1}$.

 Therefore, $c$ has to perform $|X_i|-1$ steps and eventually has to reach $d$.

 As soon as  $c$
performs  $|X_i|-1$ additional moves,
  the adversary starts the next phase $i+1$,  selecting
 $X_{i+1}$ to be the area currently  delimited by the blocked edge on one side and the passive agent on the other (a ``shift" of the former $X_i$ area in the direction of   $d$'s move); that is,
 $X_{i+1}=(X_i \setminus \{u\} )\cup \{v\}$; thus,  $| X_i |= | X_{i+1}| = x$.
 Observe that the added node $v$ could have been already visited in a previous round;
 that is    $N_i  \leq N_{i+1} \leq N_i +1$.

The adversary repeats the same process until  phase $j$ where
 $N_j= N$.
 Note that termination cannot occur before the end of this process.
In fact,  should one agent  terminate at any time during the $i-$th phase,  $i< j$ ,
then  every  attempt of the other agent to move would be blocked by the adversary,
preventing further exploration,  and the algorithm would fail on rings of size $n=N>N_i$.

 Since the explored part
  grows of at most one node in each phase   $i<j$, the total number of
  phases is at least $N- x$.  Since in each phase at least $x$ moves
  are performed by the agents, the total number of moves performed by
  the agents is at least $x (N- x) $. By choosing $x=\lceil n/2\rceil$
  the adversary ensures the claimed bound. \\

\end{proof}

\paragraph{B. Chirality and Landmark}
\ \\
Consider now the case when, instead of knowledge of an    upper bound  on the ring size, there is a landmark in the ring.

The general strategy for two agents with chirality (but without bound on the ring size)
is  essentially the same as   the one for   {\sc {\sf PT}BoundWithChirality},  where however  an agent cannot use $N$ for termination, but terminates  when it performs a complete loop around the landmark (and thus knows $n$), or  when it detects that the agents have  crossed (like in the previous algorithm).
The resulting algorithm {\sc {\sf PT}LandmarkWithChirality}
is however more efficient in terms of number of link traversals by the agents.
\begin{figure}
\footnotesize
\begin{framed}
\begin{algorithmic}

\State States: \{{\sf Init}, {\sf Bounce}, {\sf Reverse}\}.
\State $\mathit{leftSteps} \gets \bot$
\State $\mathit{rightSteps} \gets \bot$

\AtState{Init}{}
    \State \Call{Explore}{\dLeft$|$ $n$ is known: \sTerminate,  \pCatches: \sBounce}
\AtState{Bounce}{}
	\State $\mathit{leftSteps} \gets$\vEsteps
	\If{ $(\mathit{rightSteps} \neq \bot) \land  (\mathit{rightSteps} \geq \mathit{leftSteps})$ }
		\State {\sf Terminate}
	\EndIf
    \State \Call{Explore}{\dRight$|$ $n$ is known: \sTerminate, \vBtime $> 0$: \sReverse }
\AtState{Reverse}{}
	\State $\mathit{rightSteps} \gets$\vEsteps
    \State \Call{Explore}{\dLeft$|$ $n$ is known: \sTerminate, \pCatches: \sBounce}

\END
\end{algorithmic}
\caption{Algorithm {\sc {\sf PT}LandmarkWithChirality} \label{alg-ptlandmarkchirality}}

\end{framed}
\end{figure}

\begin{theorem}\label{PTlandmark}
Two agents executing Algorithm\\ {\sc {\sf PT}LandmarkWithChirality} in the {\sf PT} model  with  chirality will explore a  ring with a landmark using at most $O(n^2)$ edge traversals. Furthermore, one agent explicitly terminates, while the other  either terminates or it waits perpetually on a port.
\end{theorem}
\begin{proof}
The correctness proof is analogous to the one of Theorem \ref{PTupperbound}. Regarding complexity, the key observation is that, in a \sBounce-\sReverse phase an agent cannot do more than ${\cal O}(n)$ steps, otherwise it will loop around the landmark and terminate. Each time $\mathit{rightSteps} < \mathit{leftSteps}$ is verified by an agent,   the agent performs an additional step to the left;  if this condition is verified $2n-1$ times, that agent  has done a loop around the landmark.
Hence the   ${\cal O}(n^2)$ bound trivially follows. 
\end{proof}

Adapting the  proof  of Theorem \ref{th:lowerboundmovespt}, it is possible to show that
 the algorithm is asymptotically {\em optimal}:

\begin{theorem}
In the {\sf PT} model, any exploration algorithm for a ring of size $n$, with $n \geq 5$, with a landmark by two agents with chirality requires  $\Omega(n^2)$ edge traversals in the worst case.
\end{theorem}

\begin{proof}

Let ${\cal A}$ be a solution algorithm for exploration of a ring with a landmark
by two agents with chirality.
 Initially   the adversary locates both agents at the landmark node
   and lets them execute the algorithm until $N_1 = 5 $ nodes have been explored.
 The adversary
 operates in logical phases; let $X_i$ denote the area of the ring explored by the agents at the beginning of   phase $i$, and let $N_i=|X_i|$.
 In  phase $i$ with  $N_i <  n$, the adversary considers one agent at a time.
 If that agent, say $a$,
  tries to leave $X_i$, the adversary will block that edge
  and keeps it blocked until either $a$ changes direction
  (i.e., stays  within $X_i$)
  or also the other agent tries to leave $X_i$ (from opposite directions).
 Eventually both agents must  want to leave $X_i$ from
opposite directions:
otherwise the adversary can continue to prevent the agents
from leaving $X_i$ and
since $N_i <  n$, the algorithm ${\cal A}$ would  fail.

When the adversary determines that also the second agent, say $b$,
would want to  leave $X_i$ should it become active (recall  the first agent is  currently
waiting on a blocked edge),
the adversary
 considers what  each agent would do, according to ${\cal A}$, if  the edge
 the agent  wants to cross is  blocked indefinitely.

Clearly,  they cannot both wait indefinitely, otherwise in the ring where $n=N_i$,  the adversary,
 activating both agents at all times,
would make the algorithm fail by having the two  agents wait forever at the two endpoints of the same blocked edge, without either of them terminating (although they have explored the ring).
This means that  at least one of the agents,
after  a finite number of activations, must  change direction and
move towards a neighbouring node in  $X_i$.
Let $c\in\{a,b\}$ be such an agent, and $d$ be the other.

Upon this determination,
if $d=b$,   the adversary  activates $d$ (letting it move),
 and then makes $d$ passive; otherwise, if $d=a$, it makes $a$ passive immediately
 and let it move passively on the next node.

In any case, it also  blocks the edge from which
$c$ was trying to leave $X_i$
and keeps it blocked until $c$
performs at least $|X_i|$   moves, possibly reaching $d$.
Note that, as in the case of  Algorithm\\ {\sc {\sf PT}BoundWithChirality},
 this will always happen as shown by the   claim below.

\noindent {\bf Claim 2.}
{\em   If $c$ never reaches $d$ in this phase and it does not wait forever in the port of node $u'$ trying to reach node $u$, it will perform
an unbounded number of moves.}
 \begin{proof}

By contradiction, let $c$ satisfy the conditions of the claim, but perform only a bounded number of moves, without reaching $d$.
Thus, within finite time it will stop in a node indefinitely. When this happens, the adversary activates $d$ blocking any edge through which
$c$  wants to move. Hence no new node will be explored; since
 $N_i <  n$ the adversary would make the algorithm ${\cal A}$  fail.

\end{proof}

Further note that, during this time, agent $d$ moves (if $d=b$) or is passively
moved (if $d=a$) to a neighbour $v$  outside $X_i$.

In  {\bf Claim 2}  we   assumed that  $c$ does not wait forever in the port of node $u'$ trying to reach node $u$.
Notice that if this happen, then we can define the set $X_{i+1}=(X_i \setminus \{u\} )\cup \{v\}$, if $d$ is blocked traversing an edge to go outside $X_{i+1}$ it cannot wait forever or the adversary can create a schedule and
a configuration in which both agents are forever blocked on the endpoints of edge $(u,u')$, without the possibility of termination.
Therefore, agent $d$ eventually leaves node $v$ to go to the neighbour node in $X_{i+1}$, thus the adversary can let $c$ visit node $u$,   it then forces $c$ to sleep and it keeps $d$ blocked until $d$ leaves $v$.
At this point we are again in the situation where both agents are confined in the set of nodes $X_i=(X_{i+1} \setminus \{v\} )\cup \{u\}$.
If the agents keep behaving in this way,  the adversary makes them oscillating forever between the sets $X_i$ and $X_{i+1}$.

 Therefore, $c$ has to perform $|X_i|-1$ steps and eventually has to reach $d$.

%
%
%

 As soon as  $c$
performs  $|X_i|$  moves,
  the adversary starts the next phase $i+1$. Since $v$ is a  newly explored node,
 $X_{i+1}$  is  the former $X_i$ area augmented by node $v$; thus,
   $ N_{i+1} =N_i  + 1$.

The adversary repeats the same process until  phase $j$ when $N_i=n-1$.
 Note that termination cannot occur before the end of this process.
In fact,  should one agent  terminate at any time during the $i-$th phase,  $i< j$ ,
then  every  attempt of the other agent to move would be blocked by the adversary,
preventing further exploration,  and the algorithm would fail because the ring has not been fully explored yet.

 Since the explored part
  grows of one node in each phase, the total number of
  edge traversals is   $\sum_{i=3}^n i > {{n^2} \over 2} $
  and the  theorem follows. 
\end{proof}

\subsubsection{{\sf PT}:   Without Chirality} \label{treeagent}

Without  chirality, we have shown that two agents do not suffice (Theorem~\ref{twonotenoughpt}). We thus consider the presence of three agents in absence of chirality.
Also in this case, we need to assume some other knowledge
 because the presence of three agents  without additional information  is not sufficient (Theorem \ref{impossibilitymoreagents}) even for partial termination.

 As in the previous Section, we first assume the agents have knowledge of an upper bound $N$ on the ring size; the case when  there exists a landmark node is considered later.

 \paragraph{A.  Upper Bound}\ \\

The algorithm,  {\sc {\sf PT}BoundNoChirality},  for this case is described in Figure~\ref{alg-ptbnochirality}.
 Upon activation, at least  two of the three agents  will necessarily agree on the orientation.
 An agent   bounces only when catching  another agent, performing a ``zig-zag tour".
There are several ways in which an agent terminates.

{ More precisely,    an agent changes direction if and only if it reaches another agent that is waiting on a missing edge in the same direction.

Each agent memorizes the distance $d$ that it has traveled between the first time it changed state from {\sf Bounce}  to {\sf Reverse}. Each time the agent changes state from {\sf Bounce} to {\sf Reverse} (or viceversa) it checks that the number of steps done   is strictly greater than $d$, otherwise it terminates. If the distance is greater the agent memorizes this new distance. If the agent meets an agent on a node (condition {\sf MeetingB} and  {\sf MeetingR} of the algorithm) and distance $d$ has been set,  it also checks that the distance from the last direction change is greater than $d$ and updates $d$, otherwise it terminates. Finally, there is a trivial terminating condition of doing $N$ steps in the same direction.
\begin{figure*}
\footnotesize
\begin{framed}
\begin{algorithmic}
\State States: \{{\sf Init}, {\sf Bounce}, {\sf Reverse},{\sf MeetingR}, {\sf MeetingB}\}.

\State $d  \gets 0$

\AtState{Init}{}
    \State \Call{Explore}{\dLeft$|$ \vTnodes $\geq N$: \sTerminate,  \pCatches: \sBounce}
    
\AtState{Bounce}{}
   \State \Call{CheckD}{\vEsteps}
    \State \Call{Explore}{\dRight$|$ \vTnodes $\geq N$: \sTerminate, \pMeeting: {\sf MeetingB}, \pCatches: \sReverse }
\AtState{Reverse}
	\If{$d =0$}
	\State $d \gets$\vEsteps     \Comment{First time I change state from {\sf Bounce} to {\sf Reverse}}
	\Else
	   \State \Call{CheckD}{\vEsteps}
	\EndIf
    \State \Call{Explore}{\dLeft$|$ \vTnodes $\geq N$: \sTerminate, \pMeeting: {\sf MeetingR}, \pCatches: \sBounce}

\AtState{MeetingR}{}
	     	\If{ $\vEsteps \leq d$ }
		\State {\sf Terminate}
		\EndIf
	    \State \Call{ExploreNoResetEsteps}{\dLeft$|$ \vTnodes $\geq N$: \sTerminate, \pCatches: \sBounce} \Comment{This procedure is the same of \Call{Explore} but it does not reset \vEsteps}
	    \AtState{MeetingB}{}
	    	    	\If{ $\vEsteps \leq d$ }
		\State {\sf Terminate}
		\EndIf
	    \State \Call{ExploreNoResetEsteps}{\dRight$|$ \vTnodes $\geq N$: \sTerminate, \pCatches: \sReverse}\Comment{This procedure is the equal of \Call{Explore} but it does not reset \vEsteps}
\END
\Function{CheckD}{x}
	\If{$d >0$}
	\If{ $x \leq d$ }
		\State {\sf Terminate}
	\Else
	\State $d \gets x$
	\EndIf
	\EndIf
\EndFunction

\end{algorithmic}
\caption{Algorithm  {\sc {\sf PT}BoundNoChirality} \label{alg-ptbnochirality}}
\end{framed}
\end{figure*}

\begin{figure}
\begin{center}
\includegraphics[scale=0.45]{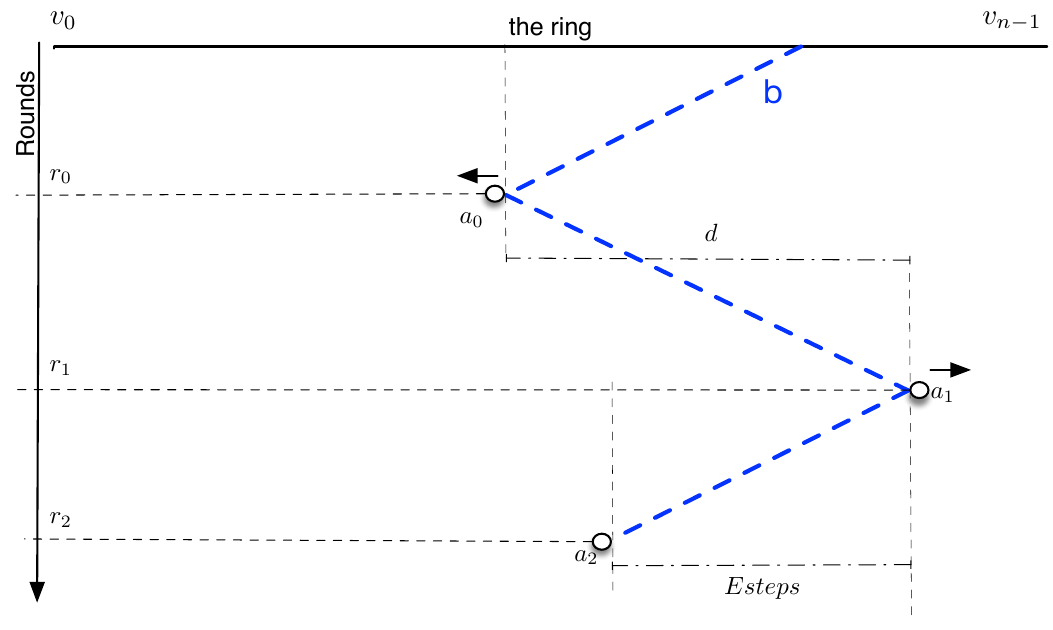}
\end{center}
\caption{ {\sf PT}NoChirality Termination  \label{figure:pt3agent}}
\end{figure}

\begin{lemma}\label{PTcorrectterm}
In  Algorithm {\sc {\sf PT}BoundNoChirality}, if an agent terminates then the ring has been explored.
\end{lemma}
\begin{proof}

The only non trivial part is the termination for Condition [{\bf if}  $\vEsteps \leq  d$ ]. Let $b$ be the first agent that terminates by satisfying condition $\vEsteps \leq  d$  (if there are several such agents, take any of them), at some round $r_2$. We will show that assuming that there are unexplored nodes in the ring at time $r_2$ leads to contradiction.

First  note that, if there are unexplored nodes at time $r_2$, no agent could have terminated before time $r_2$ by satisfying condition \vTnodes $\geq N$; i.e., $b$ is indeed the first agent to terminate.

Let us define rounds $r_0$ and $r_1$ as the last two rounds in which $b$ changed direction, with $r_0<r_1$ (see Figure \ref{figure:pt3agent}). Observe that $r_0$ and $r_1$ are well defined: At time $r_2$, $b$ has a positive value of $d$ (otherwise $\vEsteps \leq  d$ won't be satisfied). By construction of the algorithm, the first change of direction does not set $d$, only the subsequent ones do; i.e., $b$ must have had at least two changes of direction before satisfying the termination condition.

Let us define the {\em leftmost} and the {\em rightmost} agents as follows:
\begin{itemize}
  \item If at time $r_0$   agent $b$ changes direction from \dLeft to \dRight, then   agent $b$ becomes the {\em rightmost} agent, while   agent $a$ (that $b$ catched and reversed direction on) becomes the {\em leftmost} agent. Otherwise, $b$ becomes the {\em leftmost} agent and $a$ becomes the {\em rightmost} agent.
  \item The leftmost agent   remains  so unless it is overtaken by  another agent $a'$ moving left or it bounces on an agent $a''$, in which case the other agent becomes the leftmost agent.
  \item Similarly, the rightmost agent   remains  so unless it is overtaken by  another agent $b'$ moving right or it bounces on an agent $b''$, in which case the other agent becomes the rightmost agent.
  \end{itemize}

Let us define as {\em visited}  (for any round $r$, $r_0 \leq r\leq r_2$) the set of nodes which have been visited by the leftmost or the rightmost agent by round $r$ (included). A node that is not {\em visited} at round $r$ is called {\em unvisited} at round $r$. Hence, if there are no unvisited nodes, the whole ring has been explored. It is easy to show by induction on the round number, that the set of visited nodes forms a compact set of nodes between the leftmost and the rightmost agent (it starts as a single node at time $r_0$, and grows only by leftward move of the leftmost agent and by a rightward move of the rightmost agent; note that the fact that the actual leftmost/rightmost agents might change over time does not affect this). Hence, the leftmost and the rightmost visited node are well defined.

From the definition of the leftmost/rightmost agents (in particular, that the leftmost/rightmost agent never disappears, it can only be replaced by another one) it follows that as long as there are unvisited nodes (i.e. by assumption up to and including round $r_2$), the leftmost visited node is occupied by the leftmost agent and the rightmost visited node is occupied by the rightmost agent. As there are three agents altogether, this means that the internal visited nodes (different from the leftmost and the rightmost one) contain together at most one agent. As change of direction and/or test $\vEsteps \leq  d$ are performed only when an agent catches another agent, as long as there are unvisited nodes, no agent can change direction or terminate due to $\vEsteps \leq  d$ in an internal visited node.

Hence, when $b$ changed direction in round $r_1$, it was in the rightmost or leftmost visited node (depending on its direction in round $r_0$). As $b$ did not terminate in round $r_1$, $b$ set its $d$ in the last line of function \textsc{CheckD} to be the distance it traveled between $r_0$ and $r_1$. As $b$ changed direction in $r_1$, after traveling less than $d$ step it is still in an internal visited node and cannot terminate. Hence, $b$ must have terminated in round $r_2$ by satisfying $\vEsteps = d$. However, this means that the leftmost agent did not advance in any round between $r_0$ and $r_2$. In particular, it did not advance in round $r_1$, which is only possible (due to passive transport) if both   leftmost and   rightmost agents were blocked on a single edge. That is, all  the nodes have already been visited in round $r_1$, a contradiction. 
\end{proof}

\begin{theorem}\label{PTupperboundNoChirality}
Three  agents performing Algorithm\\ {\sc {\sf PT}BoundNoChirality} in the {\sf PT} model with
a known upper bound $N$ on the ring size and no chirality, explore the ring with $O(N^2)$ edge traversals. One
agent explicitly terminates,  the others  either terminate or wait perpetually on a port.
\end{theorem}
\begin{proof} The correctness of termination derives from Lemma \ref{PTcorrectterm}. It remains to prove that eventually at least one agent terminates.
Having three agents,  at least two will agree on the   same direction. We will consider this direction as global {\em left}.
It is easy to see that if an edge is perpetually removed, then eventually the agents terminate: two agents will be positioned at the end point of the missing edge and the third agent terminates detecting $\vEsteps=d$. If an agent is not forced to change direction and the edges are not perpetually removed,  then it will terminate since $\vEsteps > N$. Therefore, the adversary has to force the agents to bounce on each other. But let us notice that, as soon as an agent changes state from {\sf Bounce} to {\sf Reverse}, it sets a distance $d$;  if this distance does not increase at each state change, the agent    terminates. This implies that eventually we will have $ d > N$ and termination for $\vEsteps > N$. 

Let us now analyze the complexity of the algorithm. If an agent does not set $d$, then it performs at most $O(N)$ steps. If an agent sets $d$, its value is at most $O(N)$; there are at most $O(N)$  increases of  $d$, therefore an agent will do at most $O(N^2)$ movements. Since the number of agents is constant,  the total sum of movements over all agents is at most $O(N^2)$. 

\end{proof}


 \paragraph{B. Landmark}\ \\

 The general strategy for three agents without chirality and without an upper bound on the ring size, Algorithm  {\sc {\sf PT}LandmarkNoChirality},
is  essentially the same as   the one for Algorithm  {\sc {\sf PT}BoundNoChirality},
where however  an agent cannot use $N$ for termination, but terminates  when it performs a complete loop around the landmark (and thus knows $n$), or  when it detects that the agents have  crossed (like in the previous algorithm).
Essentially,  Algorithm \\{\sc {\sf PT}LandmarkNoChirality} is obtained by modifying Algorithm  {\sc {\sf PT}BoundNoChirality} (shown in Figure \ref{alg-ptbnochirality})
as follows: the predicate  ``$\vTnodes \geq N$" is substituted with ``$n$ is known", that is the agent has done a loop around the landmark.

\begin{theorem}\label{PTupperboundNoChiralityLandmark}
Three  agents performing Algorithm\\ {\sc {\sf PT}LandmarkNoChirality} in the {\sf PT} model with
no chirality in presence of a landmark, explore the ring with $O(n^2)$ edge traversals. One
agent explicitly terminates,  the others  either terminate or wait perpetually on a port.
\end{theorem}

The proof follows the same lines of the one of Theorem \ref{PTupperboundNoChirality}, where  termination does not happens when $\vEsteps \geq N$ but the first time an agent does a loop around the landmark.  It is easy to see that this has to happen since $d$ increases.

\subsection{{\bf Exploration in the {\sf ET}} \label{ssync:et} Model}

\subsubsection{Basic Results}

Let us first introduce a simple result on a unconscious  exploration:
\begin{theorem}\label{th:perpetual}
In the {\sf ET} model with chirality,
 two robots  can perform an unconscious exploration of the ring.
\end{theorem}
\begin{proof}
A trivial algorithm in which an agent changes direction only when it catches someone solves the exploration in {\sf ET}. 
\end{proof}

Given the previous results, a natural question  is whether there is an   algorithm with partial termination, as we have shown for the {\sf PT} model. Unfortunately the following theorem shows that, without exact knowledge of the network size, it is impossible to design such an algorithm.


\begin{theorem}\label{ETnottermination}
Let us consider the {\sf ET} model where only an upper bound on the ring size is known.
Given any number of agents, there does not exist any exploration algorithm with partial termination.
This holds true even if the ring has a landmark node, the agents have distinct IDs, there is common chirality.
\end{theorem}
\begin{proof}

The proof is by contradiction.
Let us assume that there is an exploration algorithm ${\cal A}$ that always achieves partial termination under the assumptions of the theorem.
Let us consider two different rings $R_1 =(V_1,E_1)$ and $R_2=(V_2,E_2)$, with $V_1:\{v_0,\ldots, v_{n-1}\}$ and $V_2:\{v_0,\ldots,v_{n'-1}\}$ and where $n<n'$.
Let $e_i$ denote the edge $(v_i,v_{i+1})$, where indices are taken modulo the size of the ring.
 Finally, let $N$ be the upper bound on the both ring sizes such that $n' < N$.

On ring $R_1$, starting from round $0$, the adversary perpetually removes edge $e_0$, and it schedules the activation of the agents as follows.
 In rounds  where the agents would not traverse $e_0$    or they would traverse $e_0$ from   one endpoint only, the adversary lets
all the agents be active. In the {\em busy} rounds where there are agents trying to traverse $e_0$ from both endpoints, say nodes $v_0$ and $v_{n-1}$, the adversary activates agents on $v_0$ and on $v_{n-1}$ in alternating fashion, i.e. they are never active in the same round. Notice  that $e_0$ is perpetually removed, so no one will traverse it.
Since ${\cal A}$ is assumed to be correct, it will eventually explore $R_1$. Therefore, there exists a round $r$ where ${\cal A}$ decides that $R_1$ has been explored, and at least one agent terminates.

We now consider the following schedule of activations and edges removal for ring $R_2$.
On ring $R_2$ the adversary places all agents in the same positions of ring $R_1$, i.e. they are only in nodes belonging to \\ $\{v_0,v_1\ldots,v_{n-1}\}$. In rounds when some agent wants to traverse one of edges $e_0$ and $e_n$ and no agent wants to traverse the other, the adversary will block that edge. In the {\em busy} rounds when there are agents trying to traverse both $e_0$ and $e_n$, it will alternate between making agents waiting on $e_0$ passive and blocking edge $e_n$, and making agents waiting on $e_n$ passive and blocking $e_0$.
 Notice  that, in the {\sf ET} model, such a schedule can be kept for a finite but unbounded amount of rounds, hence this can be kept until round $r$.

Observe that, until round $r$, this execution of ${\cal A}$ on $R_2$ is not distinguishable, from the point of view of agents, from the execution of ${\cal A}$ on $R_1$ previously explained. Therefore,
the agent that terminates at round $r$ in the execution of ${\cal A}$ on $R_1$ also terminates in round $r$ in the execution of ${\cal A}$ on $R_2$. However, when this happens ring $R_2$ has not been explored.


Note that the above arguments apply also if the agents have distinct IDs. If there is a landmark, it can be placed anywhere in $R_1$, and in the corresponding node in $R_2$; $R_1$ and $R_2$ will still look the same to the agents. 
\end{proof}

 Notice that Theorem \ref{ETnottermination} would hold  even if the agents were equipped with  wireless communication.

\subsubsection{Exploration Algorithm for ET}
 Algorithm {\sc {\sf ET}BoundNoChirality} is a direct adaptation of Algorithm {\sc {\sf PT}BoundNoChirality}, the only differences are that
 $N$ is set to $n-1$ (since from Theorem~\ref{ETnottermination} we know the size needs to be known precisely),
 and the  inequality check in {\sf CheckD} becomes strict:  ({\bf if} $\vEsteps < d$). As in the {\sf PT} model, three agents are employed, with no chirality assumption.
%
%
%
%
%
%

\begin{theorem}\label{th:ETpartialtermination}
Three anonymous agents performing Algorithm {\sc {\sf ET}BoundNoChirality} in the {\sf ET} model with
known ring size and no chirality explore the ring, with one agent explicitly terminating
and the other agents either terminating or waiting perpetually on a port.
\end{theorem}
\begin{proof}
Let us first observe that if an agent terminates, then it terminates correctly. The proof     follows the same steps as the one for Lemma  \ref{PTcorrectterm}. The only difference is that in {\sf ET} the {\sf CheckD} requires ({\bf if} $\vEsteps < d$), thus the part of the proof that uses the {\sf PT} assumption to handle the case $\vEsteps = d$ is not needed anymore.

What remains to be shown is that eventually at least one agent terminates. We show this by contradiction. Let us notice that if an edge is perpetually removed, an agent eventually terminates: two agents will be positioned at the two ports of the missing edge and the other one will do exactly $\vEsteps=n-1$ steps  going from one endpoint node to the others. If an edge is not perpetually removed  then, by construction (the agents bounce only on other agents, not on timeout on blocked edges) and by the {\sf ET} condition, an agent that is waiting on the port of that edge will eventually traverse it. Since the agents terminate if they traverse $n-1$ steps in the same direction, the only possibility left to consider is for the adversary to force the agents to perpetually catch each other, changing directions without increasing the $d$ value. In such a case there is a round $r^*$ after which each agent $x$ reaches a certain stable value $d_{x}$ and it always changes position at the same points $l_{x}$ and $r_{x}$; furthermore, there exists a round $r\geq r^*$ in which two agents go left and one agent goes right. In the remainder we assume the execution is in round at least $r$.

Let $a,b,c$ be the agents and let $Lxy$ indicate the event: agent $x$ catches agent $y$ at $l_x$ (and changes direction from left to right); the event $Rxy$ is analogous. Let $Dxy:D'x'y'$ denote the statement that the event $D'x'y'$ is the catching event immediately following $Dxy$.

Note that the $Dxy$ event can only be followed by an $\overline{D}xz$ or $\overline{D}zx$ events (here and in the remainder, we implicitly assume $x\neq y\neq z$), where $D\in \{L,R\}$ and $\overline{D}$ is the direction opposite to $D$. This is so  because, after $Dxy$,   agents $x$ and $z$ were moving in the same direction $\overline{D}$, $y$ remained moving in direction $D$, and only agents that move in the same direction can catch each other.
Hence, the possible sequences of direction changes of a perpetual schedule can be represented by a binary tree rooted at the initial $Dxy$ event (w.l.o.g. $Lab$ or $Lac$), with
leaves corresponding to agent termination. We call this tree \emph{Catch Tree}, see Figure \ref{figure:catchtree2}. The assumption that the algorithm does not terminate implies that there is an infinite path in this tree. In the remainder of this proof we will show that there is no such path in the \emph{Catch Tree}.

\begin{figure*}
\center
\includegraphics[width=4.6in]{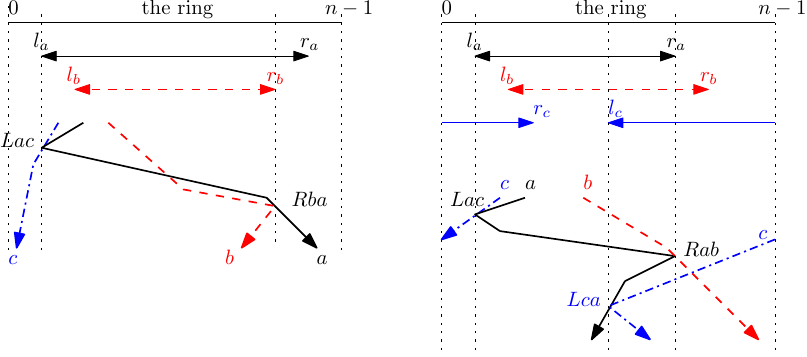}
\caption{Left: The agent ranges and events for the case $Lac:Rba$. Right: the case $Lac:Rab$.\label{fig:rangeLac}}
\end{figure*}

First, observe that the agents never meet on a node -- that would imply at least one of them performing the termination check somewhere inside its range, immediately terminating. This means that unlike the case of  Algorithm {\sc {\sf PT}BoundNoChirality} we do not have to worry about the agents overtaking each other. (They can still cross each other without noticing when traveling in opposite direction and crossing the same edge at the same time.)

Consider now a sequence of events $Dxy:\overline{D}xz:Dxy$. This corresponds to $x$ bouncing off $y$ and $z$ while those remained stationary. From the {\sf ET} condition and the fact that $y$ and $z$ are waiting on different edges (otherwise $x$ would detect $d=n-1$ and terminate) we know that eventually either $y$ or $z$ makes progress, hence this sequence cannot repeat indefinitely and eventually a different event must occur. This can be represented in the \emph{Catch Tree} by removing the bottom $Dxy$ (and its whole subtree) and adding an arrow from its parent to the top $Dxy$ (implying that the same configuration repeats). For a representation of this see Figure \ref{figure:catchtree2}.

Let $A$, $B$, $C$ be the ranges (sets of nodes an $x$ agent visits between $l_x$ and $r_x$) of agents $a$, $b$, $c$, respectively, and $\overline{A},\overline{B},\overline{C}$ be their complements.\\

\noindent {\bf Claim 3.} If the algorithm does not terminate, we have $\forall X, Y \in \{A,B,C\}, X\neq Y: \overline{X}\cap\overline{Y} = \emptyset$.
\begin{proof}
The proof is by contradiction. Assume w.l.o.g, $A$, $B$  such that $\overline{A} \cap \overline{B} \neq \emptyset$. From this and symmetry, we may assume $l_a \notin B$. This means $Lac$ is the only way for the agent $a$ to change direction while going left. Consider the possible next event (consult Figure~\ref{fig:rangeLac}, left).  It cannot be $Rba$: for that $r_b$ must be in $A$, which combined with $l_a\notin B$ implies $\overline{A}\subset \overline{B}$. Since from the moment of $Lac$ agent $a$ was moving right, starting outside $B$ and to the left of it, in order to bounce off $b$ it would need to overtake it, which is impossible.

Hence, $Lac$ must be followed by $Rab$, which in turn might be followed by $Lac$ or $Lca$. $Lac$ is the cyclic case which, as has been discussed above, cannot repeat indefinitely and can therefore be ignored.
The only possible case is $Lca$, which can only happen if $\overline{C}\subset A$.

Observe the movement of the agents (see Figure~\ref{fig:rangeLac}, right): From $Rab$ on, $b$ is moving right from $r_a$, while from $Lca$ on, $c$ is moving right from $l_c$. This means that $Rbc$ is impossible, because $c$ would have had to overtake $b$. However, $Rcb$ is also impossible as $b$ cannot get into position at $r_c$ without changing direction or crossing $\overline{B}$.    
\end{proof}

From Claim 3
it remains  to consider the case of pairwise disjoint range complements, as shown in Figure~\ref{fig:LacRba}.
\begin{figure}
\center
\includegraphics[width=3.2in]{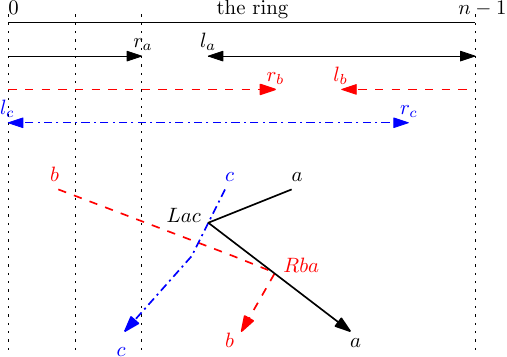}
\caption{The case of $Lac:Rba$ for disjoint range complements. \label{fig:LacRba}}
\end{figure}

\noindent {\bf Claim 4}
{\em If the algorithm never terminates, and if the agents are named in such a way  that the complement of the ranges are in order $\overline{A},\overline{B},\overline{C}$ from left to right (see Figure \ref{fig:LacRba}), then $Lac$ cannot be immediately followed by $Rba$}.

\begin{figure*}
\center
\includegraphics[width=5in]{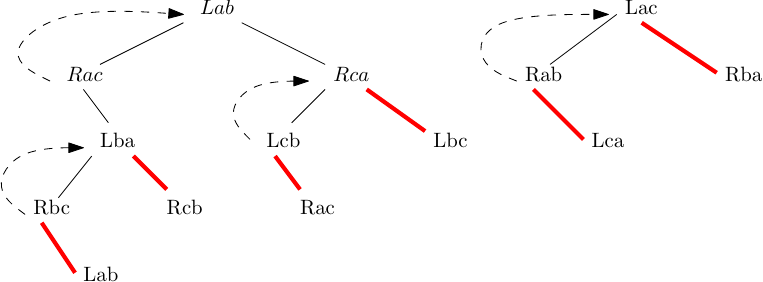}
\caption{ Catch Trees rooted at $Lab$ and $Lac$. The red edges denote forbidden sequences, while the dashed ones depict the loops. \label{figure:catchtree2}}
\end{figure*}

\begin{proof}
After event $Lac$, we have that $c$ is moving left from $l_a$. After event $Rba$, agent $b$ is moving left from $r_b$. The next event is either
$Lbc$ or $Lcb$. However, $Lbc$ is impossible, as for $c$ to reach $l_b$ without changing direction it would have to cross $\overline{C}$. Event $Lcb$ is also impossible, because $b$ would have to overtake $c$ to get into a position where $c$ could bounce on $b$.  
\end{proof}

The following corollary is immediately obtained from Claim 4 by rotation and symmetry:\\

\noindent {\bf Claim 5}
{\em If the algorithm never terminates, then the following sequences are forbidden $Lac:Rba$, $Lba:Rcb$, $Lcb:Rac$, $Rbc:Lab$, $Rca:Lbc$, $Rab:Lca$.}\\

At some point in the perpetual schedule we should have either $Lab$ or $Lac$. As can be seen in Figure~\ref{figure:catchtree2}, any branch of a Catch Tree starting from these configuration leads either to a cycle or to a forbidden sequence of events, leading to a contradiction with the assumption that the algorithm never terminates.  
\end{proof}




\noindent Note that the number of moves  performed by the agents before termination is finite but possibly unbounded.
Consider the situation when
 two agents are blocked going on opposite directions on two different edges, while the third agent goes back and forth between them;
since we are in the {\sf ET} model,   this configuration cannot be kept forever, but there is no bound on the number of rounds in which it holds.

 \section{Conclusion}

In this paper we started the investigation  of the  distributed exploration problem  for   {\em 1-interval-connected} dynamic graphs by focusing on rings  in fully synchronous and semi-synchronous environments.
We studied the impact that structural information and knowledge can  have on the solvability of the problem. In particular, we considered such factors as: knowledge of the exact size of the ring or of an upper bound, agreement on   orientation, anonymity.

 These results open the investigation of live  exploration of dynamic networks and the algorithmic study of dynamic networks in the semi-synchonous environment.

Among the several open problems, a challenging one  is the study of
live  exploration in a network of arbitrary topology where, right now, almost nothing is known. In particular,
there are no  non-trivial bounds on the number of agents,
exploration time and  type of  knowledge necessary for solvability.
Little is also known in case of dynamic networks with special but relevant underlying topologies, such as meshes, tori, hypercubes, etc.

A parallel line of research could be to provide formal proofs of that can be automatically verified by theorem provers. This seems particularly important for the case of distributed algorithms working on dynamic networks since, besides the inherent difficulties of designing a distributed algorithm, there is the additional non trivial components of considering all possible dynamic graphs.


Finally, it would be interesting to study distributed solutions to  other classical problems for mobile agents (such as gathering, in this regard see the recent \cite{gatheringsirocco}, and pattern formation)  in dynamic networks.

 \ \\

\noindent {\bf Acknowledgments.}

We would like to thank the anonymous reviewers for the comments and suggestions. Especially, for suggesting the improved bounds discussed in Section \ref{knownupp} and a correction of Theorem \ref{thm-PerpExpl}.

This research has been supported in part by  NSERC under the Discovery Grant program,
 by Dr. Flocchini's University Research Chair, and by the  VEGA 2/0165/16 grant.


\bibliographystyle{spmpsci}      
 \bibliography{mybibfile}  

\end{document}